\documentclass[a4paper,12pt]{article}
\usepackage{arxiv}
\usepackage{mathtools}

\allowdisplaybreaks
\usepackage{microtype}
\usepackage{xcolor}
\usepackage[normalem]{ulem}
\usepackage{changepage}
\usepackage{setspace}
\usepackage{amsmath, amssymb, amsthm}
\usepackage{algorithm}

\theoremstyle{definition}

\theoremstyle{plain}
\newtheorem{claim}{Claim}[section]

\begin{document}
	
\title{Gaussian Differentially Private $e$-values: Construction, Threshold Calibration, and Multiple Testing}
	
	\author{Qi Kuang}
	\author{Bowen Gang}
	\author{Yin Xia}
	
	\affil{Department of Statistics and Data Science, Fudan University}
	\date{}
	\maketitle
	
	\begin{abstract}
	This paper develops a framework for differentially private $e$-values under Gaussian differential privacy ($\mu$-GDP). We characterize the canonical noise mechanism, establishing that optimal multiplicative perturbation follows a Gaussian distribution. Using this distribution, we derive a globally sharp rejection threshold that strictly improves upon the standard Markov bound. Asymptotic analysis shows that in low-sensitivity regimes, the calibrated private test achieves a net power gain over the non-private baseline. For multiple testing, we introduce a recursive peeling algorithm that adaptively concentrates the privacy budget on the most promising hypotheses. This construction guarantees rigorous $\mu$-GDP and yields valid private $e$-values compatible with standard multiple testing procedures. Simulations and a genome-wide association study confirm that the method controls the false discovery rate while improving upon naive all-noisy privatization and recovering power close to non-private benchmarks.
	\end{abstract}
	
	
	\newpage
	\section{Introduction}
Modern statistical inference frequently relies on large-scale datasets containing sensitive individual-level records. While such data enable rigorous scientific analysis, their release poses substantial privacy risks. \citet{homer2008resolving} demonstrated that even aggregated summary statistics can be exploited to infer an individual's presence in a study, a vulnerability that compromises participant confidentiality. To mitigate information leakage, differential privacy (DP) has become the standard framework for data analysis \citep{dwork2006our,dwork2006calibrating}. Within this framework, Gaussian differential privacy ($\mu$-GDP, \cite{dong2022gdp}) serves as the preferred accounting mechanism. By framing privacy loss as a hypothesis testing problem between neighboring datasets, $\mu$-GDP provides a single-parameter characterization that admits direct statistical interpretation. This formulation yields tight composition bounds and aligns naturally with the central limit behavior of iterative algorithms, making it appropriate for complex, multi-stage analyses.

At the same time, modern large-scale inference increasingly requires statistical procedures that remain valid under sequential, adaptive, and highly dependent settings. In this context, $e$-values have emerged as a flexible alternative to classical $p$-values for statistical inference \citep{ramdas2025hypothesis}.
%The structural properties of $e$-values have established them as a flexible alternative to classical $p$-values for statistical inference \citep{ramdas2025hypothesis}. 
Formally defined as nonnegative random variables with an expectation of at most one under the null hypothesis, $e$-values confer several practical advantages. These include anytime validity \citep{safetesting}, robustness to arbitrary dependence \citep{wang2022false, xu2025bringing}, and compositional stability under both independent and dependent aggregation \citep{vovk2021values}. Such properties have driven their widespread adoption in sequential monitoring, post-hoc error control, and large-scale multiple testing. 
%Unlike classical procedures, which typically require stringent independence or positive dependence assumptions to maintain validity, e-based methods preserve error guarantees under complex, data-dependent selection mechanisms.

Despite the widespread adoption of $\mu$-GDP for privacy accounting and the growing reliance on $e$-values for inference, the construction and application of $\mu$-GDP $e$-values remain largely unexplored. Standard additive noise mechanisms violate the non-negativity and unit-expectation constraints. In this paper, we provide a systematic study of $\mu$-GDP $e$-values, progressing from optimal construction to downstream multiple testing. We characterize the noise distribution that fully exhausts the privacy budget while preserving validity, derive globally sharp rejection thresholds that strictly improve upon the standard Markov bound for single hypothesis testing, and develop a peeling algorithm that concentrates the privacy budget on the most promising hypotheses to avoid the severe noise inflation that would otherwise degrade power in large-scale multiple testing.
	%%%

	\subsection{Related Work}
	
	Recent work has begun addressing the privatization of $e$-values. \citet{csillag2026differentially} introduced a framework for transforming non-private $e$-values into differentially private counterparts under $(\epsilon, \delta)$-DP and $(\alpha, \epsilon)$-R\'enyi DP. Their approach employs a biased multiplicative noise mechanism to preserve the non-negativity constraint and the unit-expectation bound under the null hypothesis. The authors demonstrate that the resulting private $e$-values maintain statistical power and recover non-private asymptotic growth rates. Despite these contributions, their work does not provide a characterization under $\mu$-GDP, which is increasingly standard for multi-stage analyses due to its tight composition properties. Furthermore, the proposed noise mechanisms lack a derivation optimized for power maximization under a fixed privacy budget. The work relies on the standard Markov threshold $1/\alpha$ for hypothesis testing and does not address the theoretical challenges of adaptive selection or composition.
	
	Parallel developments have focused on private multiple testing using $p$-values. \citet{dwork2021differentially} proposed a differentially private Benjamini-Hochberg (BH) procedure \citep{benjamini1995controlling} that applies additive noise to log-transformed $p$-values following a randomized peeling step. This forward peeling design introduces complex dependence among the extracted statistics. Consequently, the authors derived conservative bounds for conditional false discovery rate variants rather than the standard metric. 
	
	\citet{xia2023adaptive} extended the peeling strategy to the $\mu$-GDP framework by introducing a mirror peeling algorithm. By employing a transformation that preserves the mirror-conservative property of null $p$-values, they developed an adaptive procedure that achieves finite-sample false discovery rate (FDR) control. However, establishing a precise $\mu$-GDP guarantee across adaptive peeling steps requires careful privacy accounting for the extraction of the most promising statistics, and this aspect does not appear to be fully justified in their analysis. Building upon this work, \citet{wang2025supinferableprivatemultiple} introduced the SUP framework, which utilizes a reversed peeling mechanism to maintain the super-uniformity of noisy statistics. This design yields privacy-parameter-free rejection thresholds and ensures finite-sample FDR and family-wise error rate (FWER) control. Nevertheless, the proof of the privacy guarantee in \citet{wang2025supinferableprivatemultiple} relies on a technical lemma from \citet{xia2023adaptive}, thereby inheriting the aforementioned problem. Consequently, neither framework currently provides a valid $\mu$-GDP guarantee for the selection step.

	%%%%%%%%%%%
\subsection{Our Contributions}
In this article, we present a systematic development of $\mu$-GDP $e$-values, progressing from optimal single-statistic construction to adaptive multiple testing. Our main contributions are as follows:
\begin{enumerate}
	\item \textbf{Canonical construction of $\mu$-GDP $e$-values.} We derive a multiplicative perturbation framework that transforms standard $e$-values into differentially private counterparts. Under minimal assumptions, we prove that the optimal injected noise must follow a Gaussian distribution, which induces a log-normal multiplicative factor. This characterization establishes the resulting mechanism as the canonical construction for private $e$-values under GDP.

	\item \textbf{Calibrated thresholds for single hypothesis testing.}
For $e$-values constructed under the canonical Gaussian mechanism, the standard Markov threshold $1/\alpha$ is conservative. Exploiting the known Gaussian noise distribution, we derive a globally sharp rejection threshold that preserves Type I error control at level $\alpha$. We quantify the resulting power improvement, establishing that the calibrated threshold yields a non-trivial gain in statistical power, with the most substantial recovery occurring for borderline signals near the decision boundary. Furthermore, our asymptotic power analysis reveals a counter-intuitive phenomenon: in the low-sensitivity regime, the calibrated private test can outperform the non-private baseline.
	
	\item \textbf{Improved privacy accounting for aggregated $e$-values.} We analyze the composition properties of multiple private $e$-values and show that our construction tightens the privacy bound for multiplicative aggregation. While general composition theory predicts a $\sqrt{K}\mu$ inflation for $K$ statistics, the structural properties of the independent product ensure that the combined statistic satisfies a sharper privacy guarantee.
	
		\item \textbf{A peeling algorithm for adaptive multiple testing.} We develop a recursive selection procedure that concentrates the privacy budget on the most promising hypotheses. To resolve the privacy leakage inherent in Gaussian arg-max selection, we decouple index extraction from value privatization, employing Gumbel noise for the selection step and our optimal Gaussian mechanism for the final perturbation. We also provide a private data-adaptive rule for choosing the peeling size. This procedure guarantees rigorous $\mu$-GDP across adaptive iterations and outputs a set of valid private $e$-values that can be used directly for downstream multiple testing tasks.
\end{enumerate}

	\subsection{Organization}
	The remainder of this paper is organized as follows. Section~\ref{sec:preliminaries} reviews the basic concepts of differential privacy and $e$-values. Section~\ref{sec:optimal_evalues} establishes the optimal construction of $\mu$-GDP $e$-values, alongside noise-calibrated thresholds, power analysis, and aggregation properties. Section~\ref{sec:depth_framework} develops the $\mu$-GDP e-Peeling algorithm for multiple testing. We evaluate our methodology through numerical simulations in Section~\ref{sec:experiments} and demonstrate its practical utility on a real-world GWAS dataset in Section~\ref{sec:realdata}. Finally, Section~\ref{sec:conclusion} concludes with a discussion.

	%%%%%%%%%
\section{Preliminaries}\label{sec:preliminaries}

\subsection{Differential Privacy Framework}
The formal study of differential privacy originated with the $(\epsilon, \delta)$-DP framework \citep{dwork2006our,dwork2006calibrating}. Let $\mathcal{D}$ and $\mathcal{D}'$ denote datasets that differ by a single record. A randomized mechanism $\mathcal{M}$ satisfies $(\epsilon, \delta)$-differential privacy if, for all measurable sets $S$ in the output space,
\[
\Pr(\mathcal{M}(\mathcal{D}) \in S) \leq e^{\epsilon} \Pr(\mathcal{M}(\mathcal{D}') \in S) + \delta.
\]
While this formulation provides a rigorous worst-case guarantee, the parameters $\epsilon$ and $\delta$ become difficult to track under iterative procedures. Recent developments have shifted toward a more general characterization via $f$-differential privacy \citep{dong2022gdp}, which reframes privacy loss as a hypothesis testing problem. 
Given the mechanism $\mathcal{M}$, let $P = \mathcal{M}(\mathcal{D})$ and $Q = \mathcal{M}(\mathcal{D}')$ denote the output distributions. The trade-off function $T(P, Q): [0, 1] \to [0, 1]$ is defined as
\[
T(P, Q)(\alpha) = \inf \{ \beta_{\phi} : \alpha_{\phi} \leq \alpha \},
\]
where the infimum is taken over all measurable rejection rules $\phi$, and $\alpha_{\phi} = \mathbb{E}_{P}[\phi]$ and $\beta_{\phi} = 1 - \mathbb{E}_{Q}[\phi]$ represent the Type I and Type II errors for testing the following pair of hypotheses
\begin{equation*}
	H_0: \text{The underlying dataset is } \mathcal{D} \quad \text{vs.} \quad H_1: \text{The underlying dataset is } \mathcal{D}'.
\end{equation*}
The function $T(P, Q)(\alpha)$ characterizes the minimum achievable Type II error at any significance level $\alpha$, quantifying the statistical difficulty of distinguishing $P$ from $Q$ from a single draw.
A mechanism $\mathcal{M}$ satisfies $f$-differential privacy if $T(\mathcal{M}(\mathcal{D}), \mathcal{M}(\mathcal{D}')) \geq f$ for all neighboring datasets. The classical $(\epsilon, \delta)$-DP formulation corresponds to the piecewise linear function $f_{\epsilon, \delta}(\alpha) = \max\{0, 1 - \delta - e^{\epsilon}\alpha, e^{-\epsilon}(1 - \delta - \alpha)\}$, establishing $f$-DP as a strict generalization.

Within this family, Gaussian differential privacy ($\mu$-GDP) has become the standard accounting mechanism for modern statistical procedures. It is a parametric specialization indexed by a single parameter $\mu \geq 0$. Let $\Phi$ denote the standard normal cumulative distribution function. The function $G_{\mu}(\alpha) = \Phi(\Phi^{-1}(1 - \alpha) - \mu)$ represents the trade-off function for distinguishing $\mathcal{N}(0, 1)$ from $\mathcal{N}(\mu, 1)$. A mechanism $\mathcal{M}$ satisfies $\mu$-GDP if it is $G_{\mu}$-DP. This formulation yields a \textit{general GDP composition theorem}: the sequential, possibly adaptive, application of mechanisms satisfying $\mu_k$-GDP guarantees $\sqrt{\sum_k \mu_k^2}$-GDP\citep[Corollary~3.3]{dong2022gdp}. The family of GDP guarantees aligns naturally with the central limit behavior of iterative algorithms. We adopt $\mu$-GDP as the privacy accounting framework throughout this work.

\subsection{$e$-value}
$e$-values have recently emerged as a powerful alternative to classical $p$-values for modern statistical inference.
Consider testing a null hypothesis $H_0$. An \emph{$e$-value} is a nonnegative random variable $E$ with $\mathbb{E}_{H_0}[E] \leq 1$ \citep{ramdas2025hypothesis}. By Markov's inequality, rejecting $H_0$ when $E \geq 1/\alpha$ guarantees Type~I error control at level $\alpha$. The reciprocal $1/E$ serves as a $p$-value, and any $p$-value can be transformed into an $e$-value through calibration, though some power may be sacrificed.

$e$-values serve as structural building blocks for a broad class of modern inferential procedures \citep{testingbybetting,ramdas2025hypothesis}. They underpin universal inference frameworks that yield exact confidence regions without relying on asymptotic approximations \citep{wasserman2020universal} and enable anytime-valid sequential testing, where evidence may be monitored continuously without compromising error guarantees \citep{safetesting,henzi2022valid,hao2024values}. Furthermore, they facilitate post-hoc valid inference, permitting data-dependent significance levels to be selected after observation \citep{koning2023post}. In multiple testing, they support rigorous error control under arbitrary dependence \citep{wang2022false,lee2024boosting,zhang2026generalized}. Beyond traditional statistics, their utility has recently expanded into modern machine learning, providing robust uncertainty quantification and decision making for sequential tasks \citep{zhang2025egai, gauthier2025values}.

	%%%%%%%%%

	\section{Optimal \texorpdfstring{$\mu$}{mu}-GDP $e$-values}\label{sec:optimal_evalues}
	
		With these preliminaries in place, we develop the theory of $\mu$-GDP $e$-values from the ground up. We first characterize the canonical noise mechanism for private $e$-values. Using the resulting Gaussian noise distribution, we then derive a noise-calibrated threshold for single hypothesis testing. We accompany this threshold calibration with a power analysis that separates two effects: the power recovered by replacing the standard Markov threshold with the calibrated threshold, and the net power change relative to the original non-private test. Finally, we analyze the aggregation of multiple private $e$-values to establish statistical validity and privacy composition guarantees.
	
	\subsection{Canonical Construction of $\mu$-GDP $e$-values}\label{sec:3.1}
	In the differential privacy literature, many standard mechanisms privatize a statistic through additive noise. However, direct additive perturbation is unsuitable for $e$-values, as it may violate the non-negativity and unit-expectation constraints. Following \citet{csillag2026differentially}, we employ a multiplicative perturbation framework, which corresponds to additive noise in the logarithmic domain. Let $E(\mathcal{D})$ be a valid non-private $e$-value. We define the private $e$-value as
\[
	E^{\mathrm{DP}}(\mathcal{D}) = E(\mathcal{D}) e^{-\xi},
	\]
	where $\xi$ is an independent noise variable. To derive privacy-preserving noise, we define the sensitivity of the $e$-value construction as
\[
	\Delta = \sup_{\mathcal{D} \sim \mathcal{D}'} \left| \log E(\mathcal{D}) - \log E(\mathcal{D}') \right|.
	\]
	This metric bounds the maximal change in the logarithmic scale across neighboring datasets and governs the required noise magnitude. The following theorem characterizes the exact family of noise distributions that achieve a tight $\mu$-GDP bound under natural regularity conditions.
	
	\begin{theorem} \label{thm:noise_shape}
		Consider the private $e$-value construction $E^{\mathrm{DP}}(\mathcal{D}) = E(\mathcal{D})e^{-\xi}$. Suppose the noise variable $\xi$ is symmetric with respect to its expectation and possesses a log-concave density. Further assume that the mechanism fully exhausts the privacy budget, meaning
\[
		\inf_{\mathcal{D} \sim \mathcal{D}'} T\left(E^{\mathrm{DP}}(\mathcal{D}), E^{\mathrm{DP}}(\mathcal{D}')\right)(\alpha) = G_{\mu}(\alpha), \quad \forall \alpha \in (0,1).
		\]
		Then $\xi$ must follow a normal distribution with variance $\sigma^2 = \Delta^2/\mu^2$, i.e., $\xi \sim \mathcal{N}(\tau, \Delta^2/\mu^2)$ for some $\tau \in \mathbb{R}$.
	\end{theorem}
	
	\begin{remark}
		The assumptions in Theorem~\ref{thm:noise_shape} are natural within the $\mu$-GDP framework. The symmetry condition preserves the inherent structure of the hypothesis testing interpretation of differential privacy, where distinguishing neighboring datasets is a symmetric decision problem. The log-concavity requirement ensures unimodality and regularity of the noise density, which aligns with standard statistical inference practices \citep{dong2022gdp}. The exact exhaustion condition identifies the unique noise distribution family that achieves the privacy bound.
	\end{remark}
	
	%Theorem~\ref{thm:noise_shape} identifies the Gaussian family as the unique noise distribution satisfying the privacy constraints.

	Theorem~\ref{thm:noise_shape} identifies the Gaussian family as the unique noise distribution satisfying the privacy constraints within the class of symmetric, log-concave mechanisms that exactly exhaust the $\mu$-GDP bound.
	The next theorem determines the optimal choice of $\tau$ that maximizes both statistical power and expected log-growth while ensuring $\mathbb{E}_{H_0}[E^{\mathrm{DP}}(\mathcal{D})] \leq 1$. The expected log-growth criterion is commonly used to formalize log-optimality for $e$-values \citep{kelly1956new,waudby2025universal}.

	\begin{theorem}\label{thm:optimal_expectation}
		Let $\xi \sim \mathcal{N}(\tau, \Delta^2/\mu^2)$. Under the validity constraint  $\mathbb{E}_{H_0}[ E^{\mathrm{DP}}(\mathcal{D}) ] \leq 1$, the mean parameter $\tau = \Delta^2/(2\mu^2)$ simultaneously maximizes the expected log-growth $\mathbb{E}_{H_1}[\log E^{\mathrm{DP}}(\mathcal{D})]$ and the rejection probability $\mathbb{P}_{H_1}(E^{\mathrm{DP}}(\mathcal{D}) \geq c)$ for any fixed threshold $c > 0$. 
	\end{theorem}
	In particular, taking $c=1/\alpha$ gives the power under the conventional threshold as a special case. Therefore, $\tau = \Delta^2/(2\mu^2)$ maximizes the power of the conventional-threshold test.
	
Theorems \ref{thm:noise_shape} and \ref{thm:optimal_expectation} together establish that injecting $\xi \sim \mathcal{N}(\Delta^2/(2\mu^2), \Delta^2/\mu^2)$ yields the unique log-concave, symmetric mechanism that fully utilizes the $\mu$-GDP budget while preserving validity and maximizing the testing power. Consequently, we refer to it as the canonical construction of a $\mu$-GDP $e$-value. This canonical mechanism serves as the foundation for subsequent theoretical developments, including threshold calibration, power analysis and a peeling algorithm.

	\subsection{The Noise-Calibrated Threshold}\label{sec:3.2}
	
	Given the canonical construction $E^{\mathrm{DP}}(\mathcal{D}) = E(\mathcal{D})e^{-\xi}$ with $\xi \sim \mathcal{N}(\Delta^2/(2\mu^2), \Delta^2/\mu^2)$, we now address the calibration of the rejection threshold for single hypothesis testing. The standard inferential procedure for $e$-values relies on Markov's inequality and rejects when $E^{\mathrm{DP}}(\mathcal{D}) \geq 1/\alpha$. This threshold guarantees $\mathbb{P}_{H_0}(E^{\mathrm{DP}}(\mathcal{D}) \geq 1/\alpha) \leq \alpha$ for any valid $e$-value. However, this argument uses only the $e$-value expectation bound and does not exploit the known Gaussian distribution of the privacy noise $\xi$, resulting in a conservative rejection region that unnecessarily sacrifices statistical power.
	
	%By incorporating the known Gaussian density of $\xi$, we can derive a strictly smaller threshold $c^* < 1/\alpha$ that maintains Type I error control. The following result characterizes the globally sharp threshold for the canonical $\mu$-GDP $e$-value.
	By incorporating the known Gaussian density of $\xi$, we derive a strictly smaller threshold $c^* < 1/\alpha$ that maintains Type I error control, and then analyze its consequences for statistical power.
	\begin{theorem} \label{thm:noise-calibrated_threshold}
		Let $E^{\mathrm{DP}}(\mathcal{D}) = E(\mathcal{D})e^{-\xi}$ be constructed with $\xi \sim \mathcal{N}(\Delta^2/(2\mu^2), \Delta^2/\mu^2)$. For any significance level $\alpha \in (0, 1)$, the globally sharp threshold $c^*$ satisfying $\mathbb{P}_{H_0}(E^{\mathrm{DP}}(\mathcal{D}) \geq c^*) \leq \alpha$ is given by
\[
		c^* = \begin{cases}
			\displaystyle \frac{1}{\alpha}\Phi(z^*) \exp\left(-\frac{\Delta^2}{2\mu^2} - \frac{\Delta}{\mu} z^*\right) & \text{if } \alpha \leq \Phi(z^*) \\
			\displaystyle \exp\left(-\frac{\Delta^2}{2\mu^2} - \frac{\Delta}{\mu} \Phi^{-1}(\alpha)\right) & \text{if } \alpha > \Phi(z^*)
		\end{cases}
		\]
		where $z^*$ is the unique solution to $\phi(z)/\Phi(z) = \Delta/\mu$.
	\end{theorem}
	
	The term ``globally sharp'' indicates that $c^*$ is the tightest possible constant threshold. Specifically, for any $c < c^*$, there exists a valid $e$-value $E$ such that $\mathbb{P}_{H_0}(E e^{-\xi} \geq c) > \alpha$. This calibrated threshold safely expands the rejection region, providing a direct improvement in power over the standard private Markov-threshold rule without incurring additional privacy costs or compromising error control.

	%Since $c^* < 1/\alpha$, the calibrated threshold expands the rejection region and thereby improves power relative to the standard private threshold.
	 Next, we quantify the power change through two related comparisons. The first comparison is internal to the private procedure: for the same privatized $e$-value $E e^{-\xi}$, replacing the standard Markov threshold $1/\alpha$ by the calibrated threshold $c^*$ adds the rejection region $\{c^* \leq E e^{-\xi} < 1/\alpha\}$. We refer to the probability of this event under $H_1$ as the \textit{calibration benefit}. Define the power improvement function
\[
G(x) \triangleq \mathbb{P}_{H_1}\bigl(c^* \leq x e^{-\xi} < 1/\alpha\bigr),
\]
which represents the additional rejection probability contributed by the expanded interval $[c^*, 1/\alpha)$ when the non-private $e$-value equals $x$. By the law of total expectation, the overall calibration benefit equals $\mathbb{E}_{H_1}[G(E)]$. The function $G(x)$ vanishes as $x \to 0$ and $x \to \infty$, indicating that the threshold calibration primarily affects signals near the decision boundary. The following result identifies the maximum possible value of $G(x)$ and the specific signal strength at which it occurs.
	
	\begin{proposition}\label{prop:g_max}
		For any significance level $\alpha \in (0, 1)$, the power improvement function $G(x)$ achieves its maximum at $x_{\mathrm{opt}} = \exp\left\{\Delta^2/(2\mu^2) + [\log(1/\alpha) + \log c^*]/2\right\}$. The maximum power improvement is given by
\[
		G_{\max} \triangleq G(x_{\mathrm{opt}}) = 2\Phi\left(\frac{\mu(\log(1/\alpha) - \log c^*)}{2\Delta}\right) - 1.
		\]
	\end{proposition}
	
	The number $G_{\max}$ characterizes the largest possible gain attainable among all alternative distributions, in the sense that $\mathbb{P}_{H_1}(c^* \le E e^{-\xi} < 1/\alpha) = \mathbb{E}_{H_1}[G(E)] \le G_{\max}$. As the sensitivity $\Delta$ diminishes, this bound approaches 1, demonstrating the asymptotic efficiency of the calibrated threshold.
	While $G_{\max}$ establishes a uniform ceiling, an asymptotic characterization clarifies the precise rate at which the reclaimed power decays under privacy constraints. The following theorem derives this rate for a general alternative distribution.
	
	\begin{theorem}	\label{thm:power_improvement_asymptotic}
		Let $f_E(\cdot)$ denote the probability density function of the $e$-value $E$ under the alternative hypothesis $H_1$. For any fixed privacy parameter $\mu > 0$, as $\Delta \to 0^+$, the calibration benefit satisfies
\[
		\mathbb{P}_{H_1}(c^* \leq E e^{-\xi} < 1/\alpha) \sim  \frac{f_E(1/\alpha)}{\alpha \mu} \Delta\sqrt{-2\log\Delta}.
		\]
	\end{theorem}
%	Theorem~\ref{thm:power_improvement_asymptotic} quantifies the gain at a rate of $\mathcal{O}\big(\Delta \sqrt{-\log \Delta}\big)$ from replacing the standard private threshold by the calibrated threshold. This is an improvement within the private procedure, not yet a comparison with the non-private test.

%	The second comparison is between the calibrated private test and the original non-private $e$-value test. Because the injected noise can both create and remove rejections relative to the non-private rule, this comparison requires a separate decomposition.
Theorem~\ref{thm:power_improvement_asymptotic} characterizes the power improvement obtained by calibrating the rejection threshold for the already privatized $e$-value. It does not compare the resulting private test with the original non-private test. 

Now, we are in a position to give the second comparison, between the calibrated private test and the test based on the original $e$-value.
Compared with the non-private rule $E \geq 1/\alpha$, the calibrated private rule $Ee^{-\xi}\geq c^*$ can change decisions in two directions: it may reject when the non-private rule does not, and it may fail to reject when the non-private rule rejects. These two types of decision changes yield the following decomposition of the net power difference before and after adding privacy noise:
\[
\mathbb{P}_{H_1}(E e^{-\xi} \geq c^*) - \mathbb{P}_{H_1}(E \geq 1/\alpha) = \mathbb{P}_{H_1}(E e^{-\xi} \geq c^*, E < 1/\alpha) - \mathbb{P}_{H_1}(E \geq 1/\alpha, E e^{-\xi} < c^*).
\]
The first term is the \emph{noise-induced discovery},
corresponding to alternatives that do not cross the non-private threshold but are rejected after privatization. 
The second term is the \emph{noise-induced cost},
corresponding to alternatives that would be rejected by the non-private $e$-value test but are missed after noise injection. 
The following two propositions analyze the two sides of this decomposition.
		\begin{proposition} \label{prop:gain}
    Let $f_E(\cdot)$ denote the probability density function of the $e$-value $E$ under the alternative hypothesis $H_1$. For any fixed privacy parameter $\mu>0$, as the sensitivity $\Delta\to0^+$, the noise-induced discovery satisfies
\[
    \mathbb{P}_{H_1}\left(E e^{-\xi} \ge c^* \text{ and } E < \frac{1}{\alpha}\right) \sim \frac{f_E(1/\alpha)}{\alpha \mu} \Delta\sqrt{-2\log\Delta}.
    \]
\end{proposition}

\begin{proposition}	\label{prop:power_loss}
	Under the conditions of Proposition~\ref{prop:gain}, the noise-induced cost satisfies
\[
	\mathbb{P}_{H_1}\left(E \geq \frac{1}{\alpha} \text{ and } E e^{-\xi} < c^*\right) \sim \frac{f_E(1/\alpha)}{2 e \alpha \mu^2} \cdot \frac{\Delta^2}{-\log \Delta}.
	\]
\end{proposition}
Having quantified both the noise-induced discovery and the noise-induced cost, we can now evaluate the net change in overall statistical power. Comparing the asymptotic rates in Propositions \ref{prop:gain} and \ref{prop:power_loss} establishes that the private test strictly outperforms the non-private baseline when the sensitivity is low. Specifically, as $\Delta \to 0^+$, the noise-induced cost decays at a rate of $\mathcal{O}\big(\Delta^2 / (-\log \Delta)\big)$, which is asymptotically dominated by the noise-induced discovery rate of $\mathcal{O}\big(\Delta \sqrt{-\log \Delta}\big)$. 
Hence, for sufficiently small $\Delta$, the overall power of the noise-calibrated private test strictly exceeds that of the standard non-private test, which means $\mathbb{P}_{H_1}\left(E e^{-\xi} \ge c^*\right) > \mathbb{P}_{H_1}\left(E \ge \frac{1}{\alpha}\right)$.

This rate comparison also has a direct boundary-shift interpretation that explains the apparent paradox. The private rejection condition
$E e^{-\xi} \geq c^*$ is equivalent to
$
\log E - \log(1/\alpha) \geq \zeta,
$
where $\zeta \sim \mathcal{N}\bigl(\log \Phi(z^*) - \frac{\Delta}{\mu} z^*, \frac{\Delta^2}{\mu^2}\bigr)$ and $z^*$ is as defined in Theorem~\ref{thm:noise-calibrated_threshold}. 
Note that the non-private baseline rejects $H_0$ when $\log E-\log {1/\alpha}>0$. Hence, the calibrated private procedure introduces a random shift $\zeta$ to the decision boundary. The probability that this shift lowers the rejection threshold is given by
\[
\mathbb{P}(\zeta < 0) = \Phi\left( z^* - \frac{\mu}{\Delta}\log \Phi(z^*) \right).
\]
Because $\Phi(z^*) \in (0, 1)$, the logarithmic term is strictly negative, which guarantees $\mathbb{P}(\zeta < 0) > \Phi(z^*)$. As the sensitivity $\Delta \to 0^+$, the calibration condition $\phi(z^*)/\Phi(z^*) = \Delta/\mu$ forces $z^* \to \infty$, implying $\Phi(z^*) \to 1$. Consequently, $\mathbb{P}(\zeta < 0) \to 1$. Thus, for sufficiently small $\Delta$, the random shift $\zeta$ is negative with high probability, which explains why the private test can have higher power than the non-private baseline.

\subsection{Aggregating $\mu$-GDP $e$-values}\label{sec:aggregation}

The preceding analysis focuses on a single private $e$-value. We now turn to the aggregation of multiple private $e$-values. Before establishing the general aggregation framework, it is useful to address a key design choice: why inject noise into individual $e$-values \textit{before} aggregation, rather than aggregating first and applying a privacy mechanism to the combined result? This strategy is primarily motivated by decentralized or federated settings \citep{Kaissis2020secure, kairouz2021advances}. In such environments, a fully trusted central curator who can view all raw $e$-values typically does not exist. Ensuring a formal privacy guarantee for each statistic prior to aggregation maintains robust confidentiality even when data originates from decentralized sources.

Building on this practical motivation, an inherent advantage of $e$-values is their flexibility when combining statistical evidence \citep{vovk2021values, vovk2024merging}. We first formalize the general aggregation principles. Let $\mathcal{A}$ denote any valid aggregation procedure, meaning that for any valid $e$-values $E_1, \dots, E_K$ for null hypotheses $H_0^{(1)}, \dots, H_0^{(K)}$, the output $\mathcal{A}(E_1, \dots, E_K)$ is a valid $e$-value for the intersection null $H_0 = \bigcap_{k=1}^K H_0^{(k)}$. We first state a general result for aggregation. 
\begin{proposition} \label{prop:general_composition}
	For any $\mu_k > 0$ ($k = 1, \dots, K$), if each $E_k^{\mathrm{DP}}$ is a $\mu_k$-GDP $e$-value, then the aggregated output $\mathcal{A}(E_1^{\mathrm{DP}}, \dots, E_K^{\mathrm{DP}})$ is a valid $e$-value and satisfies $\sqrt{\sum_{k=1}^K \mu_k^2}$-GDP. Specifically, if $\mu_k = \mu$ for all $k$, the aggregated output satisfies $\sqrt{K}\mu$-GDP.
\end{proposition}

A direct consequence of Proposition~\ref{prop:general_composition} is the behavior of weighted averaging under arbitrary dependence.

\begin{corollary} \label{cor:dependent_averaging}
	Let $E_1^{\mathrm{DP}}, \dots, E_K^{\mathrm{DP}}$ be valid differentially private $e$-values, all evaluated on the same dataset $\mathcal{D}$. If each $E_k^{\mathrm{DP}}$ satisfies $\mu$-GDP, then for any non-negative weights $\lambda_1, \dots, \lambda_K$ such that $\sum_{k=1}^K \lambda_k = 1$, the weighted average
\[
		E_{\mathrm{avg}}^{\mathrm{DP}}(\mathcal{D}) = \sum_{k=1}^K \lambda_k E_k^{\mathrm{DP}}(\mathcal{D})
	\]
	is a valid $e$-value and satisfies $\sqrt{K}\mu$-GDP.
\end{corollary}

As established above, combining evidence under arbitrary dependence triggers a privacy penalty that scales with $\sqrt{K}$. However, when test statistics are derived from mutually independent datasets, this standard bound is overly conservative. By aggregating through product, we derive a significantly tighter privacy guarantee that bypasses the composition barrier.

\begin{theorem} \label{thm:independent_product}
	Let $E_1, \dots, E_K$ be valid $e$-values derived from mutually independent datasets, with respective sensitivities $\Delta_1, \dots, \Delta_K$. If each private $e$-value $E_k^{\mathrm{DP}}$ is constructed via the canonical $\mu$-GDP noise mechanism described in Section~\ref{sec:3.1}, then their product $E_{\mathrm{prod}}^{\mathrm{DP}} = \prod_{k=1}^K E_k^{\mathrm{DP}}$ is a valid $e$-value and satisfies $\mu_{\mathrm{prod}}$-GDP, where
\[
		\mu_{\mathrm{prod}} = \mu \cdot \frac{\max_{1 \le k \le K} \Delta_k}{\sqrt{\sum_{k=1}^K \Delta_k^2}} \le \mu.
	\]
\end{theorem}
The privacy bound $\mu_{\mathrm{prod}}$ in Theorem~\ref{thm:independent_product} is sharp; the product does not satisfy $\mu'$-GDP for any $\mu' < \mu_{\mathrm{prod}}$. This establishes a structural advantage over general composition, which requires a privacy budget scaling as $\sqrt{K}\mu$. This sharp accounting prevents the cumulative privacy degradation that typically obscures weak signals in private aggregation.
%The privacy bound $\mu_{\mathrm{prod}}$ in Proposition~\ref{thm:independent_product} is sharp, there exists no $\mu' < \mu_{\mathrm{prod}}$ for which the product satisfies $\mu'$-GDP. This result establishes a structural advantage specific to the independent product aggregation. While the general composition rule dictates a privacy budget that inflates at a rate of $\sqrt{K}\mu$, the independent product maintains a privacy parameter strictly bounded by $\mu$. This sharp accounting prevents the cumulative privacy degradation that typically obscures weak signals in large-scale private aggregation.

	%%%%%%%

	%%%%%%%%

	\section{A Private Peeling Algorithm for $e$-values}\label{sec:depth_framework}
	Multiple testing problems arising in modern large-scale inference often involve thousands or even millions of candidate hypotheses, many of which correspond to weak and sparse signals. In such settings, preserving statistical power while ensuring rigorous privacy protection becomes particularly challenging. $e$-values are especially well suited to this regime because of their compositional properties and robustness under dependence, making them natural building blocks for private multiple testing procedures.
	
		However, applying a global $\mu$-GDP guarantee to all $m$ candidate $e$-values simultaneously incurs a severe penalty in statistical efficiency. By the general GDP composition theorem, releasing all $m$ privatized $e$-values under a fixed target privacy level $\mu$ requires assigning each coordinate a privacy budget of order $\mu/\sqrt{m}$; under the canonical Gaussian mechanism, this makes both the expectation and variance of the injected log-noise scale proportionally to $m$. This inflation rapidly obscures weak signals and degrades statistical power in sparse regimes. Adaptive peeling algorithms circumvent this limitation by iteratively selecting and perturbing only a subset of the most promising hypotheses, thereby concentrating the privacy budget.

	Data-dependent index extraction, however, introduces two distinct technical challenges. First, the sequential selection process induces dependence among the released statistics. This is where $e$-values are particularly useful: unlike classical $p$-value procedures, which typically require independence or positive dependence assumptions to maintain valid error control, $e$-value-based methods preserve their validity regardless of the dependence structure among the test statistics \citep{wang2022false,xu2025bringing}. Second, establishing precise $\mu$-GDP accounting across adaptive iterations is non-trivial. We resolve this privacy-accounting challenge by developing a decoupled peeling framework that employs Gumbel noise for stable index extraction and our canonical Gaussian mechanism for the subsequent $e$-value perturbation.

	\subsection{The Report Noisy Max Extractor}
	
	The basic building block of the peeling algorithm is a selection mechanism that isolates the most prominent $e$-value, which corresponds to the hypothesis most likely to result in a rejection. Let $\mathcal{S}$ denote the index set of $e$-values that have not yet been selected, and let $\{E_i\}_{i \in \mathcal{S}}$ be the corresponding statistics with sensitivity bounded by $\Delta$. Rather than perturbing and releasing all $e$-values, the mechanism first uses Gumbel noise to select an index. After the index is selected, it releases the corresponding $e$-value using the canonical Gaussian noise mechanism developed in Section~\ref{sec:3.1}. This approach concentrates the privacy cost on a single hypothesis per iteration. The procedure is formalized in Algorithm~\ref{alg:e_max_extractor}.
	
	\begin{algorithm}[H]
	\caption{The Report Noisy Max Extractor}
	\label{alg:e_max_extractor}
	\begin{algorithmic}[1]
		\REQUIRE $e$-values $\{E_i\}_{i \in \mathcal{S}}$ over an active candidate set $\mathcal{S}$, where the sensitivity is at most $\Delta$; privacy parameter $\mu$.
		\ENSURE Selected index $j^*$ and its differentially private $e$-value $\tilde{E}_{j^*}$.
		
		\STATE Set $\epsilon = \log \big( \Phi(\mu/(2\sqrt{2})) / \Phi(-\mu/(2\sqrt{2})) \big)$.
		\FOR{each $i \in \mathcal{S}$}
		\STATE Compute $L_i = \log E_i + g_i$, where $g_i$ is drawn independently from $\text{Gumbel}(0, 2\Delta/\epsilon)$.
		\ENDFOR
		\STATE Select the index $j^* = \arg\max_{i \in \mathcal{S}} L_i$.
		\STATE Return $j^*$ and the private $e$-value $\tilde{E}_{j^*} = E_{j^*}  e^{-\xi}$, where $\xi$ is an independent sample drawn from $\mathcal{N}(\Delta^2/\mu^2, 2\Delta^2/\mu^2)$.
	\end{algorithmic}
\end{algorithm}
	The next lemma establishes the privacy guarantee of Algorithm~\ref{alg:e_max_extractor}.
		\begin{lemma} \label{lem:e_max_extractor}
		Algorithm~\ref{alg:e_max_extractor} satisfies $\mu$-GDP. 
	\end{lemma}

\begin{remark} \label{rmk:gaussian_selection}
	The selection mechanism in \citet{xia2023adaptive} employs Gaussian noise to identify the most promising statistic. A valid privacy guarantee requires the injected noise to maintain sufficient dispersion so that the selected index does not reveal too much information about the underlying non-private statistics. As the candidate pool grows, the maximum of independent Gaussian noises becomes increasingly concentrated around its extreme-value location, which creates a privacy-accounting barrier for the claimed $\mu$-GDP guarantee. More details are provided in Appendix~\ref{app:general_selection}. 
		To address this issue, Algorithm~\ref{alg:e_max_extractor} uses Gumbel noise for index extraction. Gumbel
	perturbations are closely related to the exponential mechanism through the Gumbel-max
	trick and have been used in private selection of top-$k$ problems
	\citep{durfee2019practical}. Our use of
	Gumbel noise is different: it is not used to release a standard top-$k$ query under DP, but to decouple adaptive index selection from Gaussian value privatization under
	$\mu$-GDP. The key property is max-stability: the maximum of independent Gumbel
	random variables remains Gumbel with the same scale parameter, up to a location shift.
	Consequently, the effective scale of the selection noise does not collapse as the number
	of candidates increases, allowing the privacy loss of the selection step to remain stable
	over adaptive peeling iterations.
	
%	To address this scaling behavior, we employ Gumbel noise for index extraction. The Gumbel distribution is max-stable, meaning the distribution of the maximum retains a fixed scale parameter regardless of the candidate pool size. This property ensures that the privacy loss during selection remains invariant to the cardinality of the candidate set. 
\end{remark}

	\subsection{Recursive Peeling}
	
	Building upon the single-step extractor, we design a recursive peeling strategy to select a subset of $s$ hypotheses. In each iteration, the procedure extracts the maximum signal, injects noise, and removes the selected index from the active pool. This sequential approach concentrates the privacy budget on the $s$ most promising candidates. Consequently, it avoids the $\sqrt{m}$ noise inflation associated with evaluating all hypotheses simultaneously, preserving statistical power. The peeling algorithm is presented in Algorithm~\ref{alg:e_peeling}.
	
	\begin{algorithm}[H]
	\caption{The GDP $e$-Peeling Algorithm}
	\label{alg:e_peeling}
	\begin{algorithmic}[1]
		\REQUIRE 
		$e$-values $\{E_i\}_{i=1}^m$, 
			target peeling size $s$, 
		sensitivity $\Delta$, 
		privacy parameter $\mu$.
		\ENSURE 
		A global private $e$-value vector $\widehat{\mathcal{E}} = (\hat{E}_1, \dots, \hat{E}_m)$.
		
		\STATE Initialize active set $\mathcal{S} \gets \{1, 2, \dots, m\}$
		\STATE Initialize output vector $\widehat{\mathcal{E}} \gets (0, 0, \dots, 0)$ of length $m$
		
		\FOR{$t = 1$ to $s$}
		\STATE $(j_t^*, \tilde{E}_{j_t^*}) \gets$ \textsc{Report Noisy Max}($\mathcal{S}, \{E_i\}_{i \in \mathcal{S}}, \Delta, \mu/\sqrt{s}$)
		\STATE $\hat{E}_{j_t^*} \gets \tilde{E}_{j_t^*}$ \COMMENT{Assign the noisy value to the selected hypothesis}
		\STATE $\mathcal{S} \gets \mathcal{S} \setminus \{j_t^*\}$ \COMMENT{Remove the winner from the candidate pool}
		\ENDFOR
		\RETURN $\widehat{\mathcal{E}}$
	\end{algorithmic}
\end{algorithm}
	
		Algorithm~\ref{alg:e_peeling} allocates a privacy budget of $\mu/\sqrt{s}$ to each of the $s$ peeling iterations. By the general GDP composition theorem, the cumulative privacy cost is $\mu$.  
	%The unselected $e$-values are set to 0 rather than one. Assigning a unit value to unselected hypotheses would induce selection bias, as the conditioning on non-selection alters the null expectation. Setting them to zero preserves validity by construction.
	
		\begin{proposition}\label{prop:peeling_privacy}
			Algorithm~\ref{alg:e_peeling} satisfies $\mu$-GDP. Furthermore, each entry in $\widehat{\mathcal{E}}$ is a valid $e$-value. 
		\end{proposition}

Although the sequential, data-dependent selection mechanism induces complex dependence among the components of $\widehat{\mathcal{E}}$, this does not compromise downstream inference. Because $e$-value-based multiple testing procedures maintain validity under arbitrary dependence \citep{wang2022false, xu2025bringing}, $\widehat{\mathcal{E}}$ can be integrated directly into testing frameworks without the need for conservative post-selection corrections. 
In particular, applying the standard e-BH procedure to the private $e$-value vector $\widehat{\mathcal{E}}$ yields finite-sample FDR control.
	\begin{corollary}\label{cor:peeling_ebh_fdr}
	Let $\widehat{\mathcal{E}}$ denote the output of Algorithm~\ref{alg:e_peeling}. For any
	target level $\alpha \in (0,1)$, applying the e-BH procedure to
	$\widehat{\mathcal{E}}$ controls the FDR at level $\alpha$. Moreover, the resulting
	rejection set satisfies $\mu$-GDP.
	%	Let $\widehat{\mathcal{E}}$ be the output of Algorithm~\ref{alg:e_peeling}. For any target level $\alpha \in (0,1)$, applying the e-BH procedure to $\widehat{\mathcal{E}}$ controls the FDR at level $\alpha$. The resulting rejection set also satisfies $\mu$-GDP.
	\end{corollary}

	\subsection{Adaptive Choice of the Peeling Size}
	
		Algorithm~\ref{alg:e_peeling} requires a target peeling size $s$. This tuning parameter controls a basic trade-off. If $s$ is too large, the privacy budget allocated to each peeling iteration is $\mu/\sqrt{s}$, so the final value perturbation becomes unnecessarily large and weak signals can be washed out. If $s$ is too small, the output vector contains at most $s$ nonzero $e$-values, and the procedure can miss hypotheses that would have been rejected by a non-private e-BH procedure. Thus, the peeling size $s$ should be large enough to capture most of the true signals, yet small enough to avoid the noise inflation of simultaneous privatization.
		
		We choose $s$ by using a small part of the privacy budget to approximate the e-BH rejection count. Let $E_{(1)}\geq E_{(2)}\geq \cdots \geq E_{(m)}$ denote the ordered $e$-values. After sorting, e-BH rejects $k^*=\max\{k:E_{(k)}\geq m/(\alpha k)\}$ hypotheses. Equivalently, if $L_i=\log E_i$ and $Q_k=L_{(k)}-\log\{m/(\alpha k)\}$, then $k^*=\max\{k:Q_k\geq0\}$. Thus, a relevant quantity for choosing the peeling size is the largest $k$ for which $Q_k$ remains nonnegative. However, using this non-private cutoff directly to determine the peeling size would leak information about the data. Instead, we spend a small amount of privacy budget to release noisy versions of $Q_k$ on a grid.
		
		Define the candidate grid $\mathcal{K} = \{s_{\min}, 2s_{\min}, 4s_{\min}, \ldots\} \cap \{1, \ldots, m\}$, where $s_{\min}$ is a minimum peeling size. 
		Note that the procedure is typically insensitive to the choice of $s_{\min}$ as the grid grows geometrically; in practice, we recommend starting from a relatively small value. Since $\max_{1 \leq i \leq m}|L_i(\mathcal{D}) - L_i(\mathcal{D}')| \leq \Delta$, each order statistic $L_{(k)}$ has sensitivity at most $\Delta$. We allocate privacy budget $\mu_0$ to release
		$$\widetilde{Q}_k = Q_k + Z_k, \qquad Z_k \stackrel{\mathrm{i.i.d.}}{\sim} \mathcal{N}\left(0, \frac{|\mathcal{K}|\Delta^2}{\mu_0^2}\right), \quad k \in \mathcal{K},$$
		and reserve $\mu_{\mathrm{peel}} = \sqrt{\mu^2 - \mu_0^2}$ for the subsequent peeling step.
		If all $\widetilde{Q}_k < 0$, set $\hat{s} = s_{\min}$. Otherwise, let $\hat{k} = \max\{k \in \mathcal{K}: \widetilde{Q}_k \geq 0\}$ and let $\hat{k}^+$ be the next grid point above $\hat{k}$ (or $\hat{k}$ if it is the largest). We set
		$\hat{s}=\min\{m,\hat{k}^{+}\}$.
		The complete adaptive procedure is summarized in Algorithm~\ref{alg:adaptive_peeling_size}.
	
	\begin{algorithm}[H]
	\caption{Adaptive GDP $e$-Peeling}
	\label{alg:adaptive_peeling_size}
	\begin{algorithmic}[1]
		\REQUIRE $e$-values $\{E_1,\dots,E_m\}$, FDR level $\alpha$, minimum grid size $s_{\min}$, sensitivity $\Delta$, total privacy parameter $\mu$, budget $\mu_0$ for choosing $s$.
		\ENSURE A global private $e$-value vector $\widehat{\mathcal{E}}$.
		\STATE Set $\mathcal{K}=\{s_{\min},2s_{\min},4s_{\min},\ldots\}\cap\{1,\ldots,m\}$.
		\STATE Compute $L_i=\log E_i$ and the margins $Q_k=L_{(k)}-\log\{m/(\alpha k)\}$ for $k\in\mathcal{K}$.
		\STATE Draw independent $Z_k\sim \mathcal{N}(0,|\mathcal{K}|\Delta^2/\mu_0^2)$ and set $\widetilde Q_k=Q_k+Z_k$.
		\IF{$\max_{k\in\mathcal{K}}\widetilde Q_k<0$}
		\STATE Set $\hat{s}=s_{\min}$.
		\ELSE
		\STATE Let $\hat{k}=\max\{k\in\mathcal{K}:\widetilde Q_k\geq 0\}$.
		\STATE Let $\hat{k}^{+}$ be the next grid point above $\hat{k}$, or $\hat{k}$ if $\hat{k}$ is the largest element of $\mathcal{K}$.
		\STATE Set $\hat{s}=\min\{m,\hat{k}^{+}\}$.
		\ENDIF
		\STATE Set $\mu_{\mathrm{peel}}=\sqrt{\mu^2-\mu_0^2}$.
		\STATE Return Algorithm~\ref{alg:e_peeling} with target peeling size $\hat{s}$ and privacy parameter $\mu_{\mathrm{peel}}$.
	\end{algorithmic}
\end{algorithm}

\begin{proposition}\label{prop:adaptive_peeling}
	Algorithm~\ref{alg:adaptive_peeling_size} satisfies $\mu$-GDP. Furthermore, each entry of its output vector is a valid $e$-value.
\end{proposition}
		%%%%%%%%

	\section{Numerical Experiments}\label{sec:experiments}
	We evaluate the proposed methodology through numerical experiments. The experiments address two questions: how much the noise-calibrated threshold improves power and whether adaptive peeling mitigates the power loss from simultaneous privatization of all $e$-values in large-scale settings. 
	%We first isolate the effect of the noise-calibrated threshold in a single hypothesis test. We then assess the full multiple testing procedure under varying privacy and signal configurations, comparing fixed and adaptive peeling against the non-private e-BH procedure and a baseline that privatizes all test statistics simultaneously.
	
	%	We evaluate the proposed methodology through numerical experiments. The experiments address two questions: whether the noise-calibrated threshold recovers the power predicted by the single-test theory, and whether peeling avoids the power loss caused by privatizing all $e$-values simultaneously in large-scale multiple testing. We first isolate the effect of the noise-calibrated threshold in a single hypothesis testing setting. We then assess the performance of the full multiple testing procedure under varying privacy and signal configurations, comparing both fixed and adaptive peeling against the non-private e-BH procedure and an All-Noisy GDP baseline that privatizes all test statistics simultaneously.
	
	\subsection{Power Improvement via Noise-Calibration}\label{subsec:single_sim}
	
	We consider testing $H_0: Z \sim \mathcal{N}(0, 1)$ against $H_1: Z \sim \mathcal{N}(\lambda, 1)$ using the $e$-value $E = \exp(\lambda Z - \lambda^2/2)$ with $\lambda = \sqrt{\log(1/\alpha)}$. 
	%This choice of $\lambda$ maximizes the Chernoff bound on $\mathbb{P}_{H_0}(E \geq 1/\alpha)$ and is therefore optimal for the standard Markov threshold. 
	We fix the privacy parameter at $\mu = 0.25$ and the significance level at $\alpha = 0.05$. The sensitivity varies over $\log_{10}(\Delta) \in \{-3, -2.75, \dots, 1\}$. For each configuration, we estimate empirical Type I error and power using $10^5$ Monte Carlo trials.
	
	Figure~\ref{fig:calibration_log_delta} displays the results. Both the standard Markov threshold $1/\alpha$ and the calibrated threshold $c^*$ maintain Type I error control at the nominal level. The calibrated threshold consistently achieves higher power than the standard private threshold across all sensitivity values. Over the low-to-moderate sensitivity range, approximately $-3 \leq \log_{10}(\Delta) \leq -0.5$, the calibrated procedure also exceeds the non-private baseline. The gray line marks the non-private baseline, whose power is exactly 0.5 in this simulation.
	 This strict improvement aligns with the asymptotic analysis in Section~\ref{sec:3.2}.
	
	As sensitivity increases, the power advantage of the calibrated threshold becomes more pronounced. For $\log_{10}(\Delta) > 0.25$, the power under the standard $1/\alpha$ threshold drops to near zero, while the calibrated threshold retains positive power. At $\log_{10}(\Delta) = 0$, the calibrated threshold yields approximately five times the power of the standard method. The power ratio plot in Figure~\ref{fig:calibration_log_delta} (right panel) highlights this gain; the apparent divergence at higher sensitivities reflects the vanishing denominator as the standard procedure loses all power. These results confirm that exploiting the known noise distribution to calibrate the rejection threshold effectively recovers statistical power and, under low sensitivity, yields a net gain relative to the non-private procedure.

	\begin{figure}[ht]
		\centering
		\includegraphics[width=0.95\textwidth]{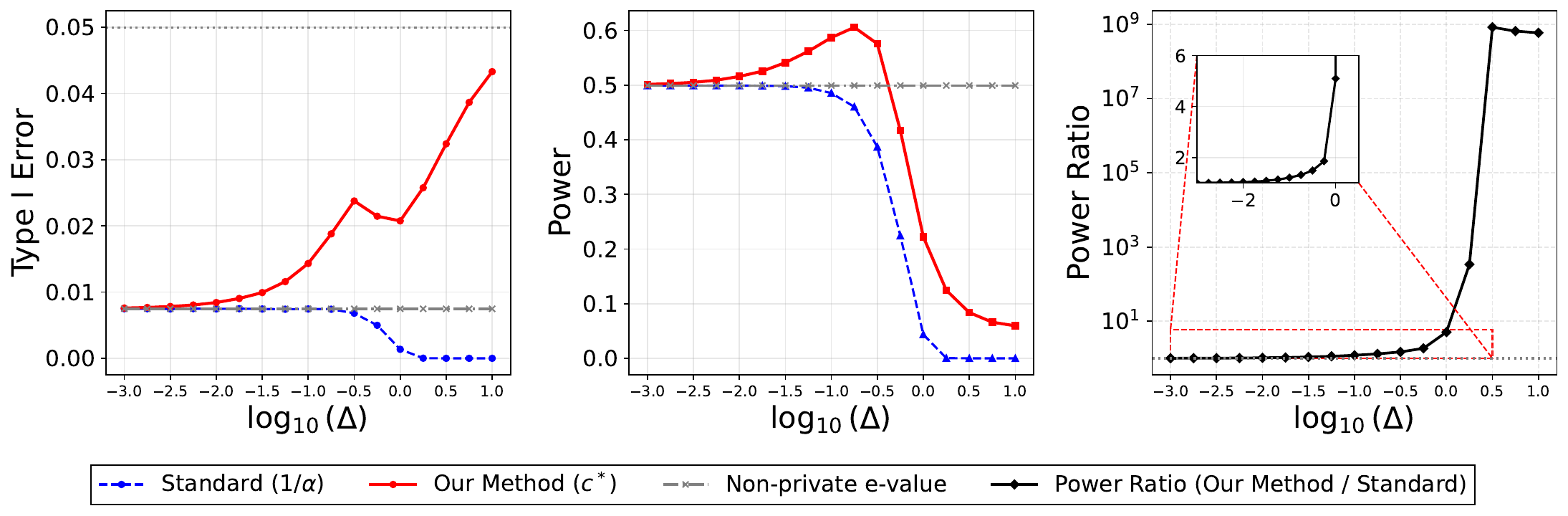}
		\caption{Left: Type I error across varying sensitivity $\Delta$. Middle: Power across varying sensitivity $\Delta$. Right: Power ratio demonstrating the improvement of our method over the standard threshold.}
		\label{fig:calibration_log_delta}
	\end{figure}
	
	\subsection{Multiple Testing for Independent $e$-values}\label{sec:5.2}
	
We evaluate the peeling algorithms from Section~\ref{sec:depth_framework} in a multiple testing setting with FDR control.
We consider $m$ simultaneous hypothesis tests, with $m_1$ true alternatives and $m-m_1$ true nulls. For each hypothesis $i \in \{1, \dots, m\}$, we generate an independent observation
$
X_i \sim \mathcal{N}(\eta_i, 1),
$
where the mean $\eta_i$ is set to $\eta_{\mathrm{alt}} > 0$ for the first $m_1$ hypotheses (signals) and to $0$ for the remaining $m-m_1$ hypotheses (nulls). The non-private $e$-values are constructed as
\[
E_i = \exp\!\left(\lambda X_i - \frac{\lambda^2}{2}\right), \qquad \lambda = \sqrt{\log(m/\alpha)}.
\]
%This choice of $\lambda$ maximizes the Chernoff bound on the tail probability $\mathbb{P}_{H_0}(E_i \geq 1/\alpha)$, making it a standard benchmark for $e$-value construction.
	We compare the following four procedures:
	\begin{itemize}[itemsep=0pt, topsep=2pt, parsep=0pt, partopsep=0pt]
		\item \textbf{Non-private e-BH.} The standard e-BH procedure \citep{wang2022false} applied to the non-private $e$-values $\{E_i\}_{i=1}^m$. This serves as a baseline without injected noise.
		
		\item \textbf{All-Noisy GDP.} A baseline differentially private method that first constructs $E_i^{\text{DP}}=E_i e^{-\xi_i}$ where $\xi_i$ are i.i.d. random variables from $\mathcal{N}\big(\frac{m\Delta^2}{2\mu^2}, \frac{m\Delta^2}{\mu^2}\big)$. The standard e-BH procedure is subsequently applied to $\{E^{\text{DP}}_i\}_{i=1}^m$. It follows from the general GDP composition theorem that this procedure satisfies $\mu$-GDP.

		\item \textbf{Fixed peeling.} We first apply Algorithm~\ref{alg:e_peeling} with a fixed peeling size $s$ to obtain private $e$-values $\hat{\cal{E}}$, and then apply the e-BH procedure to $\hat{\cal{E}}$.
		
		\item \textbf{Adaptive peeling.} We apply Algorithm~\ref{alg:adaptive_peeling_size} to select the peeling size privately and then run the peeling-based e-BH procedure with the remaining privacy budget.
		
	%	The algorithm from Section~\ref{sec:depth_framework}, which first extracts a subset of $m$ promising candidates using the Gumbel-based Report Noisy Max mechanism, privatizes the selected $e$-values with the canonical Gaussian noise, and then applies the e-BH procedure to the output vector.
	\end{itemize}
				Unless otherwise stated, we fix the following default configuration. The total number of hypotheses is set to $m = 10^5$, with $m_1 = 100$ true signals of strength $\eta_{\mathrm{alt}} = 4.0$. The target FDR level is $\alpha = 0.05$, and the fixed peeling size is $s = 500$. For adaptive peeling, we allocate $\mu_0=0.1\mu_{\mathrm{gdp}}$ to the noisy margin step used to choose the peeling size, and use $\mu_{\mathrm{peel}}=\sqrt{\mu_{\mathrm{gdp}}^2-\mu_0^2}$ for the subsequent peeling procedure. Since GDP composes through squared privacy parameters, this leaves $\mu_{\mathrm{peel}}=\sqrt{0.99}\mu_{\mathrm{gdp}}\approx0.995\mu_{\mathrm{gdp}}$ for peeling. We take the candidate set $\mathcal{K}$ to be the dyadic grid starting at $s_{\min}=50$. The final choice of $\hat{s}$ follows Algorithm~\ref{alg:adaptive_peeling_size}. Privacy parameters are set to $\mu_{\mathrm{gdp}} = 4\epsilon/\sqrt{10\log(1/\delta_{\mathrm{dp}})}$  with $\epsilon = 0.5$ and $\delta_{\mathrm{dp}} = 10^{-3}$, as in \cite{xia2023adaptive}.
	The sensitivity of the $e$-values is set to $\Delta = 5 \times 10^{-3}$.

To study the impact of each parameter, we vary one parameter at a time while holding the others at their default values. The sensitivity grid is $\log_{10}(\Delta) \in \{-4, -3.75, \dots, -1\}$, the sparsity grid is $m_1 \in \{50, 100, 150, 200, 250, 300\}$, the signal strength grid is $\eta_{\mathrm{alt}} \in \{2, 3, 4, 5, 6, 7\}$, and the privacy budget grid is $\epsilon \in \{0.1, 0.2, \dots, 1.5\}$. For each configuration, we report the empirical FDR and average power (AP), which is defined as the expected proportion of true discoveries among the non-null hypotheses:
\[
\mathrm{AP} = \mathbb{E}\left[ \frac{\sum_{i \in \mathcal{S}_1} \delta_i}{m_1} \right],
\]
where $\mathcal{S}_1$ denotes the index set of false null hypotheses, $m_1 = |\mathcal{S}_1|$, and $\delta_i = \mathbb{I}\{H_i \text{ is rejected}\}$. In simulations, this quantity is estimated by the empirical average of the true positive rate across replications.
The results are averaged over $100$ independent Monte Carlo trials.

%%%%%%%%%

	\begin{figure}[htbp]
		\centering
		\includegraphics[width=0.85\textwidth]{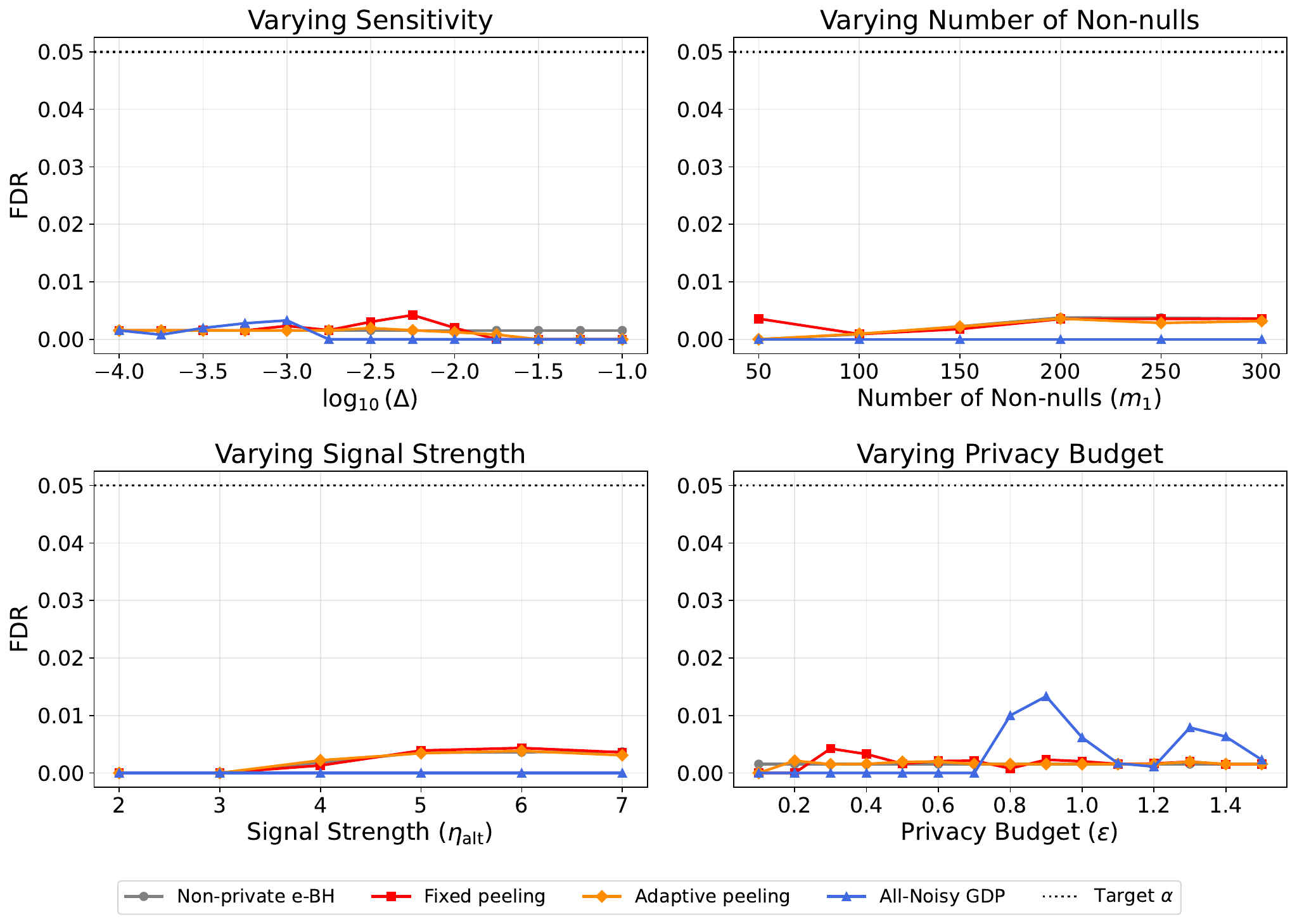}
		\caption{Independent case; empirical FDR for varying $\Delta$, $m_1$, $\eta_{\mathrm{alt}}$, and $\epsilon$, averaged over repeated simulations. The horizontal line is the target level $\alpha=0.05$.}
		\label{fig:inde_fdr}
	\end{figure}
	
	\begin{figure}[htbp]
		\centering
		\includegraphics[width=0.85\textwidth]{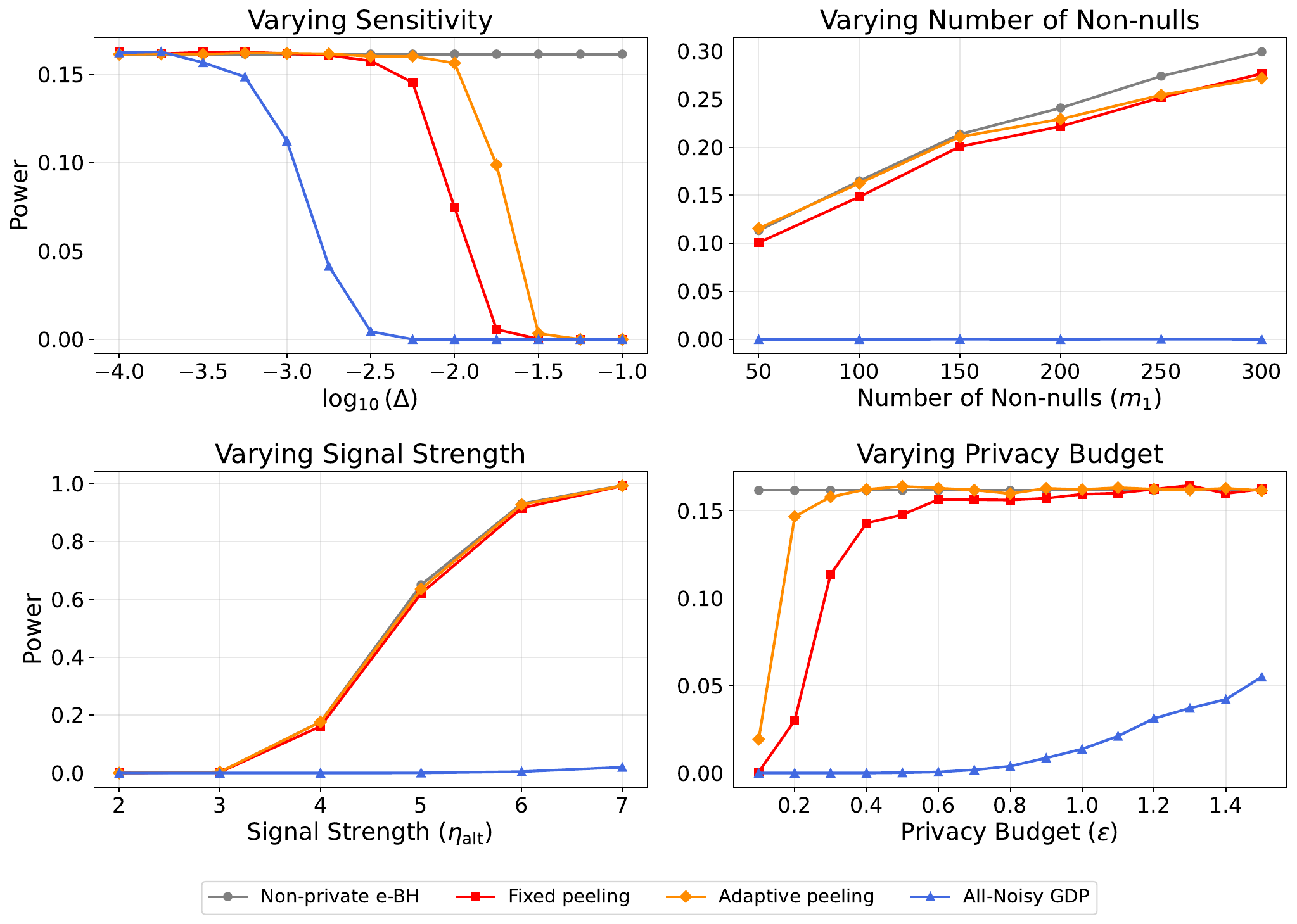}
		\caption{Independent case; empirical power for varying $\Delta$, $m_1$, $\eta_{\mathrm{alt}}$, and $\epsilon$, averaged over repeated simulations.}
		\label{fig:inde_power}
	\end{figure}
	
Figures \ref{fig:inde_fdr} and \ref{fig:inde_power} summarize the empirical performance across all parameter grids. All methods control the FDR at or below the target level $\alpha = 0.05$, confirming the theoretical guarantees. The All-Noisy GDP baseline has little power in most cases, reflecting the cost of protecting all $m$ $e$-values simultaneously. Both peeling variants are substantially more powerful, with adaptive peeling outperforming fixed peeling.

%with adaptive peeling oFixed peeling performs well when the pre-specified value $s=500$ is appropriate for the rejection frontier, while adaptive peeling achieves comparable power without requiring this value to be chosen in advance.

	%%%%%%%%%%%%%
%	Figures \ref{fig:inde_fdr} and \ref{fig:inde_power} show three main patterns.
%	First, all methods control FDR below the target level across all grids.
%	Second, our method consistently dominates the All-Noisy baseline in power, with the largest gains under high sensitivity (large $\Delta$) and strong privacy (small $\epsilon$).
%	Third, as expected, power increases with denser signals ($n_1$ larger) and stronger effects ($\eta_{\mathrm{alt}}$ larger), and our method remains close to non-private e-BH in most regimes.
	
	\subsection{Correlated $e$-values}
	Now we consider the case where the $e$-values are correlated. In practice, test statistics frequently exhibit dependence due to shared latent factors or structured noise. To evaluate the robustness of our framework under such conditions, we replace the independent data generation model with a one-factor dependence structure:
\[
X_i = \eta_i + \sqrt{\rho}\, W + \sqrt{1-\rho}\, Z_i,
\]
	where $W \sim \mathcal{N}(0, 1)$ is a common latent factor and $Z_i \overset{\text{i.i.d.}}{\sim} \mathcal{N}(0, 1)$ are idiosyncratic errors. We set the correlation parameter to $\rho = 0.3$. Under this specification, the resulting $e$-values inherit positive dependence through the shared component $W$. All other experimental configurations remain identical to the independent setting, including the mean configuration $\eta_i$, the $e$-value construction, the parameter grids, and the baseline comparison methods. 
 The empirical results are summarized in Figures \ref{fig:depen_fdr} and \ref{fig:depen_power}.

	\begin{figure}[htbp]
		\centering
		\includegraphics[width=0.85\textwidth]{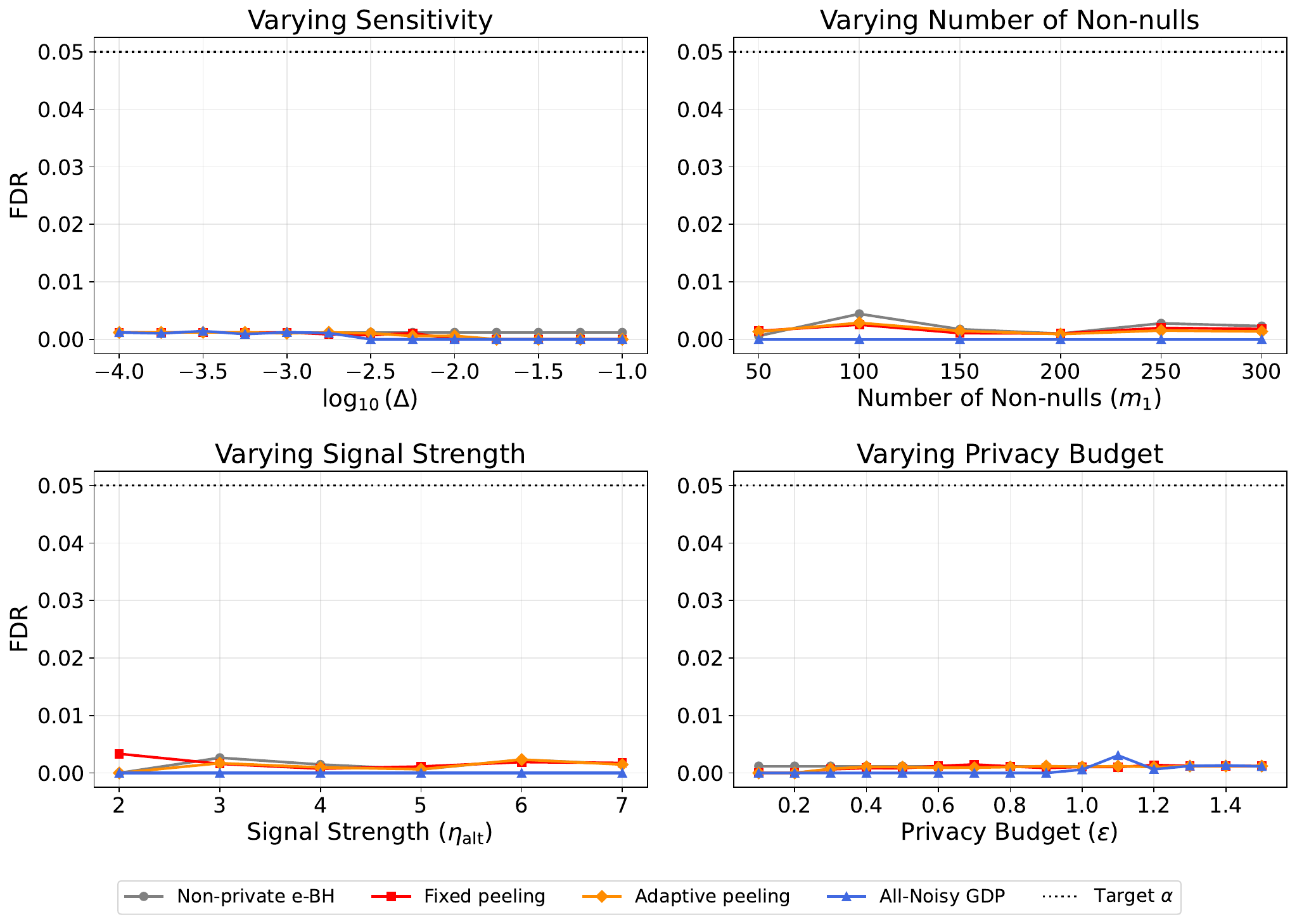}
		\caption{Correlated case; empirical FDR for varying $\Delta$, $m_1$, $\eta_{\mathrm{alt}}$, and $\epsilon$, averaged over repeated simulations. The horizontal line is the target level $\alpha=0.05$.}
		\label{fig:depen_fdr}
	\end{figure}
	
	\begin{figure}[htbp]
		\centering
		\includegraphics[width=0.85\textwidth]{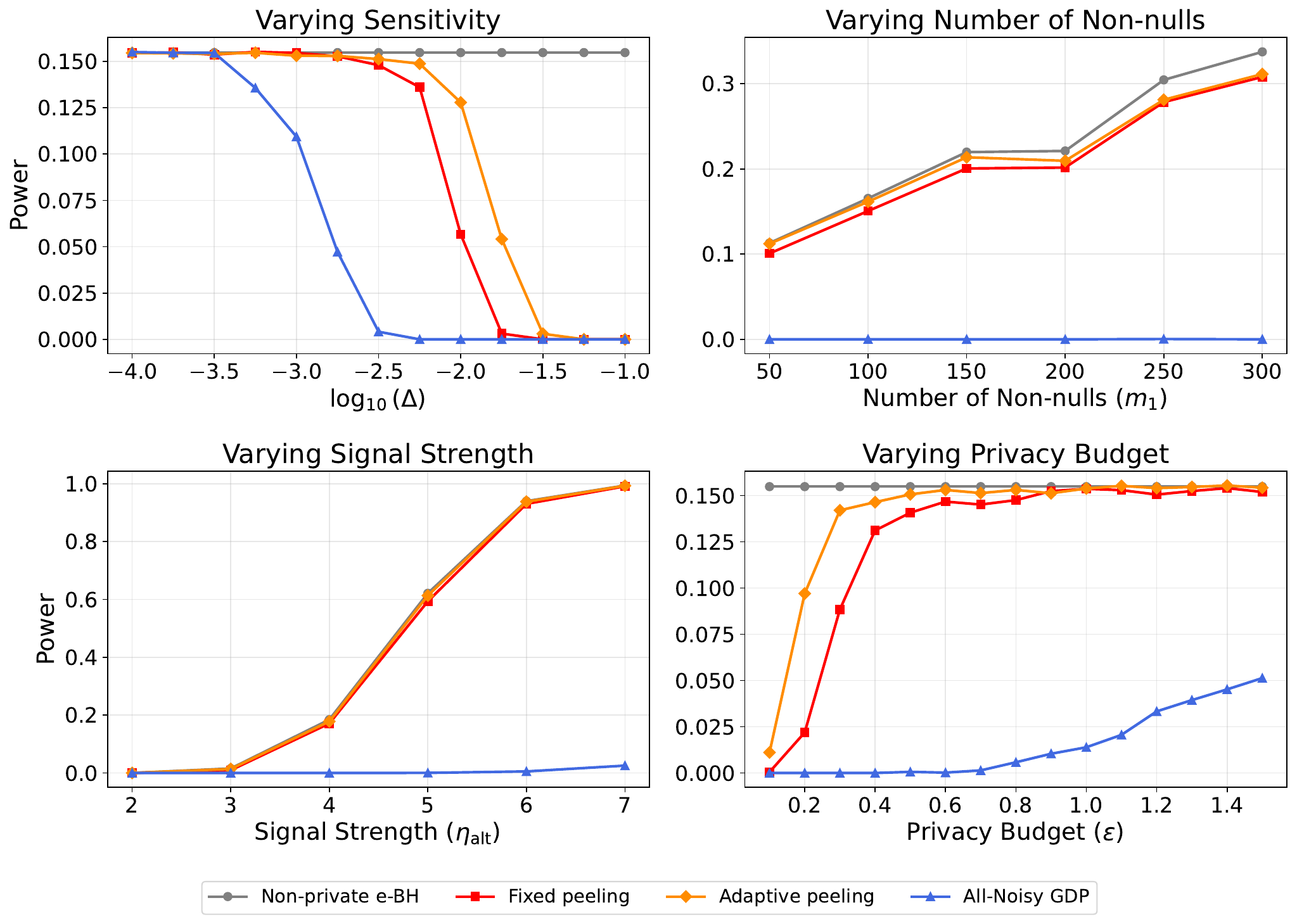}
		\caption{Correlated case; empirical power for varying $\Delta$, $m_1$, $\eta_{\mathrm{alt}}$, and $\epsilon$, averaged over repeated simulations.}
		\label{fig:depen_power}
	\end{figure}
	
The figures display patterns consistent with the independent setting. All methods maintain FDR control at the target level, confirming that our procedure remains valid under both data-dependent peeling and correlation-induced dependence. In terms of power, the peeling methods uniformly dominate the All-Noisy GDP baseline, with adaptive peeling outperforming fixed peeling.
%Adaptive peeling closely tracks fixed peeling across most regimes and can be more favorable when the fixed choice of $s$ leads to excessive noise. Overall, across both independent and correlated settings, the results show rigorous FDR control with materially improved power relative to the naive All-Noisy GDP approach.

	\section{Real Data Analysis}\label{sec:realdata}
	We illustrate the practical utility of the proposed methodology using large-scale genomic data. Summary statistics for Systemic Lupus Erythematosus (SLE) were obtained from the OpenGWAS database (Study ID: ukb-b-3454), \url{https://opengwas.io/datasets/ukb-b-3454}. The dataset comprises $m = 6,196,160$ single nucleotide polymorphisms (SNPs).
	
	\citet{homer2008resolving} demonstrated that GWAS summary statistics do not provide perfect privacy protection: an individual's participation in a study can be inferred with high confidence from such aggregate data. To address this vulnerability, we apply the proposed GDP peeling framework, including the fixed-size peeling algorithm (Algorithm~\ref{alg:e_peeling}) and its adaptive peeling-size version (Algorithm~\ref{alg:adaptive_peeling_size}), to ensure privacy guarantees while maintaining statistical utility.
	
	For each hypothesis, we construct the $e$-value as $E_i = \exp(\lambda Z_i - \lambda^2/2)$, where $Z_i$ is the reported $z$-score and $\lambda = \sqrt{\log(m/\alpha)}$. We apply both the fixed and adaptive GDP $e$-Peeling procedures to identify significant SNPs across a range of pre-specified FDR levels $\alpha \in \{0.01, 0.015, \dots, 0.05\}$.
	
			The privacy parameter is set to $\mu = 0.25$, following \citet{xia2023adaptive}. The sensitivity parameter is $\Delta = 5 \times 10^{-3}$. We compare the four methods mentioned in Section~\ref{sec:5.2}. For fixed peeling, the peeling size is set to 500. For adaptive peeling, we allocate $\mu_0=0.1\mu$ to the noisy margin step and use $\mu_{\mathrm{peel}}=\sqrt{\mu^2-\mu_0^2}\approx0.995\mu$ for the subsequent peeling procedure. The dyadic grid starts at $s_{\min}=50$. The results are displayed in Figure~\ref{fig:realdata}.
	
	\begin{figure}[ht]
		\centering
		\includegraphics[width=0.8\textwidth]{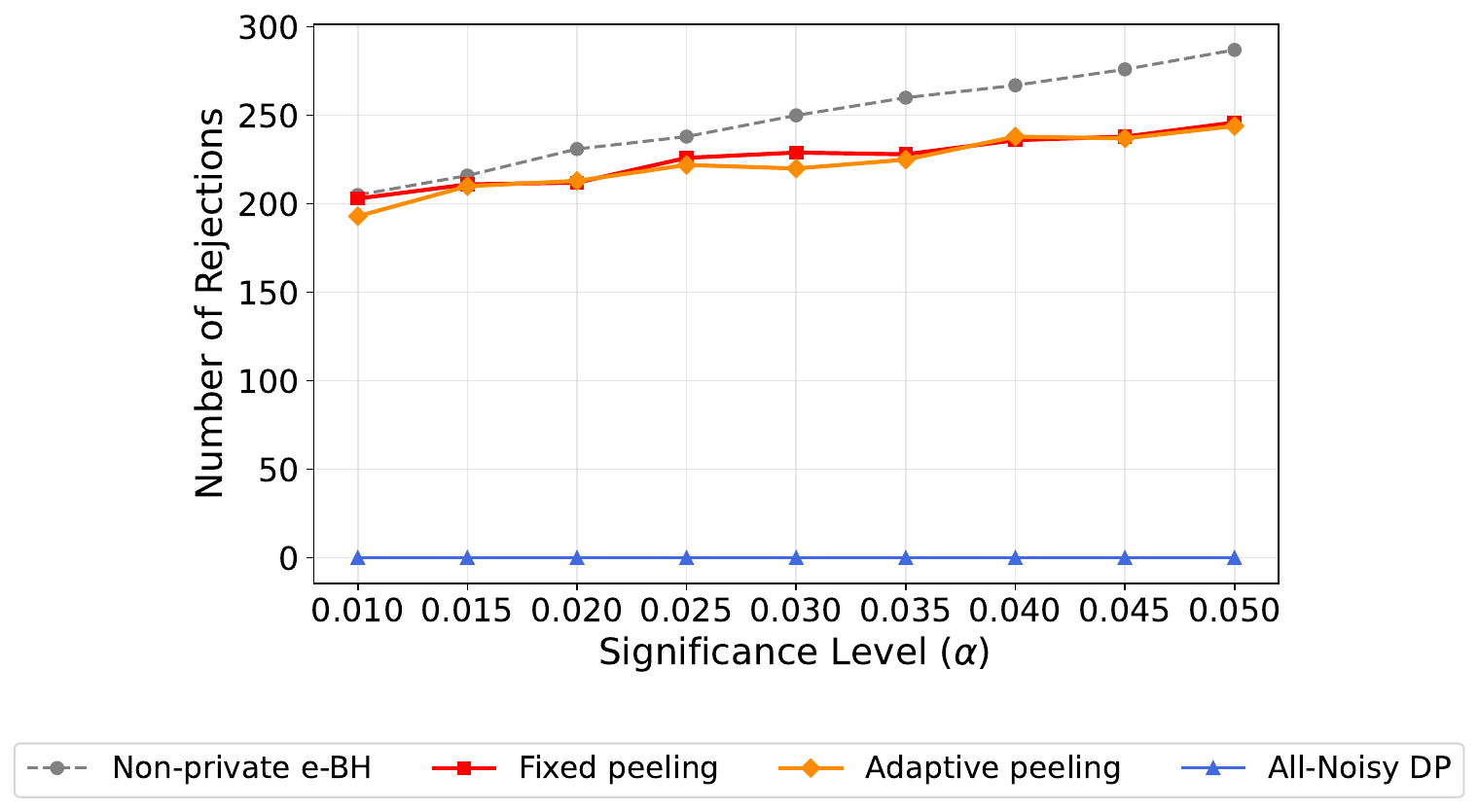}
		\caption{GWAS dataset; number of discoveries for non-private e-BH, All-Noisy GDP, fixed peeling, and adaptive peeling across varying FDR levels.}
		\label{fig:realdata}
	\end{figure}
	
	Figure~\ref{fig:realdata} shows that the All-Noisy GDP baseline yields no discoveries, as the noise magnitude required to protect $m$ hypotheses overwhelms the underlying signals. In contrast, both peeling algorithms recover a large fraction of the non-private e-BH discoveries. Fixed peeling performs well on this dataset, while the adaptive procedure uses Algorithm~\ref{alg:adaptive_peeling_size} to choose the peeling size from the data and gives comparable discovery counts without requiring this tuning parameter to be fixed in advance.

	\section{Discussion}\label{sec:conclusion}

	This paper establishes a comprehensive framework for differentially private inference based on $e$-values. We characterized the canonical $\mu$-GDP $e$-value construction, derived a noise-calibrated threshold that expands the rejection region without compromising Type~I error control, and provided a precise asymptotic analysis of the resulting power gain and unavoidable loss. Furthermore, we developed a decoupled peeling algorithm that leverages the robustness of $e$-values to arbitrary dependence, enabling rigorous multiple testing under strict privacy constraints.
	
	A natural theoretical extension of our work concerns the optimality of the Gaussian mechanism under relaxed privacy constraints. Theorems \ref{thm:noise_shape} and \ref{thm:optimal_expectation} establish the canonical construction under the exact exhaustion condition. One may ask whether this construction remains optimal when the privacy requirement is weakened to a lower bound. Specifically, suppose the noise variable $\xi$ is symmetric with respect to its expectation and possesses a log-concave density. Consider the class of all noise distributions satisfying
\[
	\inf_{\mathcal{D} \sim \mathcal{D}'} T\left(E^{\mathrm{DP}}(\mathcal{D}), E^{\mathrm{DP}}(\mathcal{D}')\right)(\alpha) \geq G_{\mu}(\alpha), \quad \forall \alpha \in (0,1).
	\]
	Does the Gaussian choice $\xi \sim \mathcal{N}(\tau, \Delta^2/\mu^2)$ strictly maximize statistical power among all mechanisms in this class? 
	
	Standard intuition from differential privacy suggests that tighter privacy constraints uniformly degrade statistical utility. Under this view, the boundary solution that exactly achieves $G_{\mu}$ would naturally dominate any strictly stronger mechanism. However, the asymptotic analysis in Section~\ref{sec:3.2} reveals a more complex relationship. Injecting noise can yield a net power gain over the non-private baseline in low-sensitivity regimes. This result demonstrates that a more stringent privacy constraint does not monotonically reduce testing power. Consequently, the monotonicity intuition fails, and establishing whether the Gaussian mechanism remains optimal under the relaxed inequality constraint is non-trivial. We leave this characterization as an open problem.

	\bibliographystyle{apalike}
	\bibliography{reference.bib}

	\newpage

	\begin{center}
	\begin{spacing}{2}
		{\Large Supplementary Material for ``Gaussian Differentially Private $e$-values: Construction, Threshold Calibration, and Multiple Testing''}
	\end{spacing}
\end{center}
This supplement provides technical proofs of the theoretical results in the main paper
(Section~\ref{proof}) and gives a detailed analysis of a privacy-accounting gap in the
selection mechanism of \cite{xia2023adaptive} (Section~\ref{app:general_selection}).
\appendix

\section{Technical Proofs}\label{proof}

\subsection{Proof of Theorem~\ref{thm:noise_shape}}
\begin{proof}
	By the symmetry assumption, there exists a mean $\tau \in \mathbb{R}$ such that the centered noise $\xi_0 \triangleq \xi - \tau \stackrel{d}{=} \tau-\xi$. Let $F(x)$ be the CDF of $\xi_0$. The trade-off function between $E^{\mathrm{DP}}(\mathcal{D}')$ and $E^{\mathrm{DP}}(\mathcal{D})$ is:
\begin{align*}
		T\big(E^{\mathrm{DP}}(\mathcal{D}'), E^{\mathrm{DP}}(\mathcal{D})\big) &= T\big(\log E(\mathcal{D}') - \xi, \log E(\mathcal{D}) - \xi\big) \\
		&= T\big(\xi-\tau-\log E(\mathcal{D}'), \xi-\tau - \log E(\mathcal{D})\big)\\
		&= T\big(\xi-\tau, \xi-\tau + \log E(\mathcal{D}') - \log E(\mathcal{D})\big)\\
		&= T\big(\xi_0, \xi_0 + \log E(\mathcal{D}') - \log E(\mathcal{D})\big).
	\end{align*}
	Let $t \triangleq \log E(\mathcal{D}') - \log E(\mathcal{D})$, so  $|t| \le \Delta$.
	
	Since $\xi$ is log-concave, by Proposition A.3 of \citet{dong2022gdp}, for $t > 0$, the trade-off function is $F(F^{-1}(1-\alpha) - t)$. Because $F$ is increasing, this is minimized over $t \in (0, \Delta]$ at $t = \Delta$. Let $T_F(\alpha) \triangleq F(F^{-1}(1-\alpha) - \Delta)$.
	
	For $t < 0$, let $s = -t \in (0, \Delta]$. The trade-off is $T(\xi_0, \xi_0 - s) = T(-\xi_0, -\xi_0 + s)$. By symmetry $\xi_0 \stackrel{d}{=} -\xi_0$, this is equal to $T(\xi_0, \xi_0 + s)$. Since $s > 0$, we can apply Proposition A.3 of \citet{dong2022gdp} again to obtain {\color{black}$T(\xi_0, \xi_0 + s)=F(F^{-1}(1-\alpha) - s)$. }This is minimized over $s \in (0, \Delta]$ at $s = \Delta$ (which corresponds to $t = -\Delta$). This yields the same minimum curve $T_F(\alpha)$.
		Thus, $\inf_{\mathcal{D} \sim \mathcal{D}'} T = G_\mu$ simplifies to:
	
\begin{equation} \label{eq:global_identity}
		T_F(\alpha) = G_\mu(\alpha), \quad \forall \alpha \in (0, 1).
	\end{equation}
	Before solving \eqref{eq:global_identity}, we verify that $0 < F(x) < 1$ for all $x \in \mathbb{R}$ and $F(x)$ is strictly increasing. This ensures $x = F^{-1}(1-\alpha)$ can take every value in $\mathbb{R}$.
	
	Suppose $F(A) = 0$ for some $A \in \mathbb{R}$. Let $x_* = \sup\{x \in \mathbb{R} : F(x) = 0\}$. As $\alpha \to 1$, $F^{-1}(1-\alpha) \to x_*$. Because $\Delta > 0$, we have $F^{-1}(1-\alpha) - \Delta < x_*$ for $\alpha$ sufficiently close to $1$. By the definition of $x_*$, this yields $T_F(\alpha) = F(F^{-1}(1-\alpha) - \Delta) = 0$, which contradicts $G_\mu(\alpha) > 0$ for $\alpha \in (0, 1)$.
	
	Suppose $F(B) = 1$ for some $B \in \mathbb{R}$. Let $x^* = \inf\{x \in \mathbb{R} : F(x) = 1\}$. Since $F(B)=1$, $x^*$ is finite. As $\alpha \to 0$, $F^{-1}(1-\alpha) \to x^*$, yielding $\lim_{\alpha \to 0} T_F(\alpha) = F(x^* - \Delta)$. Since $x^* - \Delta < x^*$, the definition of $x^*$ implies $F(x^* - \Delta) < 1$. This contradicts $\lim_{\alpha \to 0} G_\mu(\alpha) = 1$.
	
	Since $F(x)$ cannot reach $0$ or $1$ at any $x \in \mathbb{R}$, the support of the log-concave density $p(x)$ for $\xi_0$ must be $\mathbb{R}$. Thus, $p(x) > 0$ for all $x \in \mathbb{R}$, ensuring $F(x)$ is strictly increasing.

	Expanding \eqref{eq:global_identity} yields $F(F^{-1}(1-\alpha) - \Delta) = \Phi(\Phi^{-1}(1-\alpha) - \mu)$. Let $x = F^{-1}(1-\alpha)$. Since $0 < F(x) < 1$ for all $x \in \mathbb{R}$ and $F(x)$ is strictly increasing, $x$ can be any real number. Applying $\Phi^{-1}$ gives $\Phi^{-1}(F(x - \Delta)) = \Phi^{-1}(F(x)) - \mu$. Let $H(x) = \Phi^{-1}(F(x))$, we obtain:
\begin{equation*}
		H(x) - H(x-\Delta) = \mu, \quad \forall x \in \mathbb{R}.
	\end{equation*}
	The general solution is $H(x) = (\mu/\Delta)x + P(x)$, where $P(x)$ is $\Delta$-periodic.

	The density can be expressed as $p(x) = F'(x) = \phi(H(x))H'(x)$. Taking the logarithm yields the log-density:
\begin{equation*}
		V(x) = \log p(x) = -\frac{1}{2}H(x)^2 + \log H'(x) - \frac{1}{2}\log 2\pi.
	\end{equation*}
	Using the periodicities $H(x+n\Delta) = H(x) + n\mu$ and $H'(x+n\Delta) = H'(x)$ for any integer $n \in \mathbb{Z}$, we evaluate $V(x)$ at $x+n\Delta$:
\begin{align*}
		V(x+n\Delta) &= -\frac{1}{2}\big(H(x)+n\mu\big)^2 + \log H'(x) - \frac{1}{2}\log 2\pi \\
		&= V(x) - n\mu H(x) - \frac{1}{2}n^2\mu^2.
	\end{align*}
	Because $\xi$ is log-concave, $V(x)$ is a concave function on $\mathbb{R}$. A property of concave functions is that their forward differences are monotone decreasing. Thus, for any $\delta > 0$ and any $a < b \in \mathbb{R}$, applying this property to the points $a+n\Delta$ and $b+n\Delta$ yields:
\begin{equation*}
		V(a+\delta+n\Delta) - V(a+n\Delta) \ge V(b+\delta+n\Delta) - V(b+n\Delta).
	\end{equation*}
	Substituting the expansion of $V(x+n\Delta)$ into both sides:
\begin{equation*}
		\big(V(a+\delta) - V(a)\big) - n\mu \big(H(a+\delta) - H(a)\big) \ge \big(V(b+\delta) - V(b)\big) - n\mu \big(H(b+\delta) - H(b)\big).
	\end{equation*}
	Rearranging to group the terms multiplied by $n\mu$ gives:
\begin{equation*}
		n\mu \Big( \big(H(b+\delta) - H(b)\big) - \big(H(a+\delta) - H(a)\big) \Big) \ge \big(V(b+\delta) - V(b)\big) - \big(V(a+\delta) - V(a)\big).
	\end{equation*}
	This must hold as $n \to \pm \infty$, meaning $H(b+\delta) - H(b) = H(a+\delta) - H(a)$. By an easy argument we can prove that if we define $A(x) = H(x)-H(0)$, we have $A(x+\delta) - A(x) = A(\delta)$ for all $x \in \mathbb{R}$ and $\delta > 0$. Since $A(x)$ is continuous, we can conclude that $A(x) = cx$ for some $c$. Meanwhile, since $H(x) = (\mu/\Delta)x + P(x)$, we have $A(x) = H(x) - H(0) = (\mu/\Delta)x + P(x) - P(0)$. Thus, $P(x) - P(0) = (c - \frac{\mu}{\Delta})x$. Since $P(x)$ is $\Delta$-periodic, we have $P(\Delta) - P(0) = 0$, which gives $c = \mu/\Delta$.

	Therefore, we have $H(x) = (\mu/\Delta)x + C_0$.
	Since $H(x) = \Phi^{-1}(F(x))$, we obtain that:
\begin{equation*}
		F(x) = \Phi\big((\mu/\Delta)x + C_0\big).
	\end{equation*}
	Since $\xi_0 \stackrel{d}{=} -\xi_0$, we have $F(0) = 0.5$, so we have $C_0 = 0$. Thus $F(x)$ is the CDF of $\mathcal{N}\left(0, \Delta^2/\mu^2\right)$. Hence $\xi = \xi_0 + \tau \sim \mathcal{N}\left(\tau, \Delta^2/\mu^2\right)$.
\end{proof}

\subsection{Proof of Theorem~\ref{thm:optimal_expectation}}
\begin{proof}
	Since the injected noise $\xi \sim \mathcal{N}\left(\tau, \Delta^2/\mu^2\right)$ is independent of the dataset $\mathcal{D}$, we have:
\begin{equation*}
		\mathbb{E}_{H_0}[E^{\mathrm{DP}}(\mathcal{D})] = \mathbb{E}_{H_0}[E(\mathcal{D})] \cdot \mathbb{E}[e^{-\xi}].
	\end{equation*}
	Given that $E(\mathcal{D})$ is a valid $e$-value ($\mathbb{E}_{H_0}[E(\mathcal{D})] \le 1$), $E^{\mathrm{DP}}(\mathcal{D})$ is a valid $e$-value if and only if $\mathbb{E}[e^{-\xi}] \le 1$.
	Using the moment-generating function of $\xi$, this requires:
\begin{equation*}
		\mathbb{E}[e^{-\xi}] = \exp\left(-\tau + \frac{\Delta^2}{2\mu^2}\right) \le 1 \iff \tau \ge \frac{\Delta^2}{2\mu^2}.
	\end{equation*}
	Under the alternative hypothesis $H_1$, the expected log-growth is:
\begin{equation*}
		\mathbb{E}_{H_1}[\log E^{\mathrm{DP}}(\mathcal{D})] = \mathbb{E}_{H_1}[\log E(\mathcal{D})] - \mathbb{E}[\xi] = \mathbb{E}_{H_1}[\log E(\mathcal{D})] - \tau.
	\end{equation*}
	This objective is strictly decreasing in $\tau$.
	The rejection probability at a fixed threshold $c>0$ is:
\begin{align*}
		\mathbb{P}_{H_1}\left(E^{\mathrm{DP}}(\mathcal{D}) \ge c\right) &= \mathbb{P}_{H_1}\left(E(\mathcal{D})e^{-\xi} \ge c\right) \\
		&= \mathbb{E}_{H_1} \left[ \mathbb{P}\left(\xi \le \log E(\mathcal{D})-\log c \mid \mathcal{D}\right) \right] \\
		&= \mathbb{E}_{H_1} \left[ \Phi\left( \frac{\log E(\mathcal{D})-\log c - \tau}{\frac{\Delta}{\mu}} \right) \right].
	\end{align*}
	Since the standard normal CDF $\Phi(\cdot)$ is strictly increasing, the expected rejection probability is also decreasing in $\tau$ for every fixed $c>0$. Since both the expected log-growth and the rejection probability are decreasing in $\tau$, their maxima over the feasible region $\tau \ge \Delta^2/(2\mu^2)$ are achieved at the lower bound $\tau = \Delta^2/(2\mu^2)$.
\end{proof}

\subsection{Proof of Theorem~\ref{thm:noise-calibrated_threshold}}

\begin{proof}
Denote $\sigma=\Delta/\mu$ and $\tau=\sigma^2/2$. For a fixed threshold $c$, define
\[
	g_c(x)=\mathbb{P}(xe^{-\xi}\ge c)
	=\Phi\left(\frac{\log x-\log c-\tau}{\sigma}\right),
	\qquad x>0.
\]
Then
\[
	\mathbb{P}_{H_0}(Ee^{-\xi}\ge c)=\mathbb{E}_{H_0}[g_c(E)].
\]
The goal is to find the smallest threshold $c$ such that $\mathbb{E}_{H_0}[g_c(E)]\le \alpha$ for every valid $e$-value $E$.
Define
\[
	S(c)=\sup_{x>0}\frac{g_c(x)}{x}.
\]
We then have $g_c(x)\le S(c)x$ for all $x>0$. Hence
\[
	\mathbb{E}_{H_0}[g_c(E)]\le S(c)\mathbb{E}_{H_0}[E]\le S(c).
\]
Thus, if $S(c)\le \alpha$, the threshold $c$ is automatically valid for every valid $e$-value.
To compute $S(c)$, put
$z=\frac{\log x-\log c-\tau}{\sigma}$ and $x=c\exp(\tau+\sigma z).$
Then
\[
	\frac{g_c(x)}{x}
	=
	\frac{\Phi(z)}{c\exp(\tau+\sigma z)}.
\]
Differentiating shows that the unique maximizer is the point $z^*$ satisfying
$
	\frac{\phi(z^*)}{\Phi(z^*)}=\sigma.
$
Therefore
\[
	S(c)=\frac{1}{c}\Phi(z^*)\exp(-\tau-\sigma z^*),
	\qquad
	x_c^*=c\exp(\tau+\sigma z^*),
\]
where $x_c^*$ is the value of $x$ at which $g_c(x)/x$ is maximized; equivalently, equality holds in $g_c(x)\le S(c)x$ at $x=x_c^*$.

The solution to 
$S(c)=\alpha$ is
\[
	c_1=\frac{1}{\alpha}\Phi(z^*)\exp(-\tau-\sigma z^*),
\]
and the corresponding maximizing value is
\[
	x_{c_1}^*=\frac{\Phi(z^*)}{\alpha}.
\]
Thus $c_1$ is always a valid threshold. The remaining question is whether it is threshold. This depends on whether the value $x_{c_1}^*$ can be used to construct a valid worst-case $e$-value. A two-point construction with mass $1/x_{c_1}^*$ at $x_{c_1}^*$ and the remaining mass at zero is valid only when $x_{c_1}^*\ge 1$. This is exactly the condition $\alpha\le \Phi(z^*)$. When $x_{c_1}^*<1$, the threshold $c_1$ remains valid but is no longer sharp (we will show this later), and a different argument is needed to identify the sharp threshold. We now discuss the two regimes separately.

\textit{Regime 1: \(\alpha\le \Phi(z^*)\).}
Here $x_{c_1}^*\ge 1$. Validity follows directly from $S(c_1)=\alpha$:
\[
	\mathbb{E}_{H_0}[g_{c_1}(E)]\le \alpha\mathbb{E}_{H_0}[E]\le \alpha.
\]
For sharpness, suppose that a smaller threshold $c<c_1$ were valid for all $e$-values. Define
\[
	\mathbb{P}(E^*=x_{c_1}^*)=\frac{1}{x_{c_1}^*},
	\qquad
	\mathbb{P}(E^*=0)=1-\frac{1}{x_{c_1}^*}.
\]
This is a valid $e$-value because $x_{c_1}^*\ge 1$ and $\mathbb{E}[E^*]=1$. When $E^*=0$, rejection is impossible; when $E^*=x_{c_1}^*$, the rejection probability is $g_c(x_{c_1}^*)$. Therefore, for this $E^*$,
\[
	\mathbb{P}(E^*e^{-\xi}\ge c)
	=
	\frac{g_c(x_{c_1}^*)}{x_{c_1}^*}
	>
	\frac{g_{c_1}(x_{c_1}^*)}{x_{c_1}^*}
	=
	S(c_1)
	=
	\alpha,
\]
where the strict inequality follows from $c<c_1$, since $g_c(x)=\Phi((\log x-\log c-\tau)/\sigma)$ is strictly decreasing as a function of the threshold $c$. Thus no smaller threshold can be valid. This proves the validity and sharpness of the first branch,
\[
	c^*=c_1=\frac{1}{\alpha}\Phi(z^*)\exp(-\tau-\sigma z^*),
	\qquad \alpha\le \Phi(z^*).
\]

\textit{Regime 2: \(\alpha>\Phi(z^*)\).}
Now $x_{c_1}^*<1$. The threshold $c_1$ is still valid, but we can not use the sharpness argument as before. To find the sharp threshold, first note that any threshold valid for all $e$-values must in particular be valid for the deterministic $e$-value $E\equiv 1$. For this $e$-value, the rejection probability is $\mathbb{P}(e^{-\xi}\ge c)=g_c(1)$. Therefore any valid threshold must satisfy $g_c(1)\le \alpha$. Since $g_c(1)=\Phi((-\log c-\tau)/\sigma)$ is strictly decreasing in $c$, the smallest threshold passing this necessary condition is obtained by imposing
\[
	g_c(1)=\alpha.
\]
Writing $y=\Phi^{-1}(\alpha)$ gives the candidate
\[
	c_2=\exp\{-\tau-\sigma y\}
	=\exp\{-\tau-\sigma\Phi^{-1}(\alpha)\}.
\]
We now prove that this candidate threshold $c_2$ is valid for every valid $e$-value.
For this value, the corresponding maximizing value is
\[
	x_{c_2}^*=c_2\exp(\tau+\sigma z^*)=\exp\{\sigma(z^*-y)\}<1.
\]
The last inequality follows because we are in the regime $\alpha>\Phi(z^*)$, so $y=\Phi^{-1}(\alpha)>z^*$ and hence $\sigma(z^*-y)<0$.
Moreover, by the formulas for $S(c)$ and $x_c^*$ derived above,
\[
	S(c_2)
	=
	\frac{1}{c_2}\Phi(z^*)\exp(-\tau-\sigma z^*)
	=
	\frac{\Phi(z^*)}{c_2\exp(\tau+\sigma z^*)}
	=
	\frac{\Phi(z^*)}{x_{c_2}^*}.
\]
Let $S=S(c_2)$ and define
\[
	\widetilde g(x)=
	\begin{cases}
		Sx, & 0\le x\le x_{c_2}^*,\\
		g_{c_2}(x), & x>x_{c_2}^*.
	\end{cases}
\]
This auxiliary function is chosen for two simple reasons. First, it lies above $g_{c_2}$: for $x\le x_{c_2}^*$ this follows from the definition of $S$, and for $x>x_{c_2}^*$ it equals $g_{c_2}$. Second, $\widetilde g$ is increasing and concave. On $[0,x_{c_2}^*]$ this is clear because $\widetilde g(x)=Sx$ with $S>0$. For $x>x_{c_2}^*$, write $z=(\log x-\log c_2-\tau)/\sigma$. Then
\[
	g_{c_2}'(x)=\frac{\phi(z)}{\sigma x}>0,
	\qquad
	g_{c_2}''(x)=-\frac{\phi(z)}{\sigma^2x^2}(z+\sigma).
\]
Since $x=c_2\exp(\tau+\sigma z)$ is increasing in $z$, the condition $x>x_{c_2}^*$ implies $z>z^*$. Also, the defining equation for $z^*$ gives $\sigma=\phi(z^*)/\Phi(z^*)$. To see that $z^*+\sigma>0$, define $q(z)=z\Phi(z)+\phi(z)$. Then $q'(z)=\Phi(z)>0$ and $\lim_{z\to-\infty}q(z)=0$, so $q(z)>0$ for all $z$. Hence
\[
	z+\frac{\phi(z)}{\Phi(z)}
	=
	\frac{q(z)}{\Phi(z)}
	>0,
\]
and in particular $z^*+\sigma=z^*+\phi(z^*)/\Phi(z^*)>0$. Therefore $z+\sigma>0$ whenever $z>z^*$, and the displayed formula for $g_{c_2}''(x)$ gives $g_{c_2}''(x)<0$ for $x>x_{c_2}^*$. Finally,
\[
	g_{c_2}'(x_{c_2}^*)=\frac{\phi(z^*)}{\sigma x_{c_2}^*}
	=\frac{\Phi(z^*)}{x_{c_2}^*}
	=S,
\]
where the last equality uses the preceding display. Thus the left and right derivatives agree at $x_{c_2}^*$. Hence $\widetilde g$ has a nonnegative and nonincreasing derivative, which makes it increasing and concave. Since $x_{c_2}^*<1$, the point $1$ lies in the second branch of $\widetilde g$, so $\widetilde g(1)=g_{c_2}(1)$. Therefore, for every valid $e$-value,
\[
	\mathbb{E}_{H_0}[g_{c_2}(E)]
	\le \mathbb{E}_{H_0}[\widetilde g(E)]
	\le \widetilde g(\mathbb{E}_{H_0}[E])
	\le \widetilde g(1)
	=g_{c_2}(1)
	=\alpha.
\]
Here the second inequality follows from Jensen's inequality, and the third uses both $\mathbb{E}_{H_0}[E]\le 1$ and the monotonicity of $\widetilde g$. Thus $c_2$ is valid.

It remains to prove sharpness. If $c<c_2$, then the deterministic valid $e$-value $E^*\equiv 1$ gives
\[
	\mathbb{P}(E^*e^{-\xi}\ge c)=g_c(1)>g_{c_2}(1)=\alpha.
\]
Therefore no smaller threshold can be valid, and the second branch is
\[
	c^*=c_2=\exp\{-\tau-\sigma\Phi^{-1}(\alpha)\},
	\qquad \alpha>\Phi(z^*).
\]

Substituting $\sigma=\Delta/\mu$ and $\tau=\Delta^2/(2\mu^2)$ gives the two formulas in the theorem.
\end{proof}

\subsection{Proof of Proposition~\ref{prop:g_max}}
\begin{proof}
	Let $\sigma = \Delta/\mu$. By definition, the power improvement function for a given signal $E = x$ is the probability that the noise-injected $e$-value falls into the calibration gap:
\begin{equation*}
		G(x) \triangleq \mathbb{P}_{H_1}\left( c^* \le x e^{-\xi} < \frac{1}{\alpha} \right).
	\end{equation*}
	Taking the logarithm inside the probability, we obtain:
\begin{equation*}
		G(x) = \mathbb{P}_{H_1}\left( \log c^* \le \log x - \xi < \log(1/\alpha) \right) = \mathbb{P}_{H_1}\left( \log x - \log(1/\alpha) < \xi \le \log x - \log c^* \right).
	\end{equation*}
	Since $\xi \sim \mathcal{N}(\sigma^2/2, \sigma^2)$, standardizing $\xi$ yields the difference of two standard normal cumulative distribution functions:
\begin{equation*}
		G(x) = \Phi\left( \frac{\log x - \log c^* - \sigma^2/2}{\sigma} \right) - \Phi\left( \frac{\log x - \log(1/\alpha) - \sigma^2/2}{\sigma} \right).
	\end{equation*}
	To find the maximum, we take the derivative of $G(x)$ with respect to $y = \log x$ and set it to zero:
\begin{equation*}
		\frac{\partial G}{\partial y} = \frac{1}{\sigma} \phi\left( \frac{y - \log c^* - \sigma^2/2}{\sigma} \right) - \frac{1}{\sigma} \phi\left( \frac{y - \log(1/\alpha) - \sigma^2/2}{\sigma} \right) = 0,
	\end{equation*}
	where $\phi(\cdot)$ is the standard normal density. Because $\phi(u)$ is symmetric and strictly decreasing for $u > 0$, $\phi(u) = \phi(v)$ with $u \neq v$ implies $u = -v$. Thus, we have:
\begin{equation*}
		\frac{y - \log c^* - \sigma^2/2}{\sigma} = - \frac{y - \log(1/\alpha) - \sigma^2/2}{\sigma}.
	\end{equation*}
	Solving for $y$ yields the optimal log-score $y_{opt} = \sigma^2/2 + [\log(1/\alpha) + \log c^*]/2$. Exponentiating this gives $x_{opt} = \exp\big( \Delta^2/(2\mu^2) + [\log(1/\alpha) + \log c^*]/2 \big)$. 
	
	Substituting $y_{opt}$ back into the argument of the first $\Phi$ function yields:
\begin{equation*}
		u = \frac{y_{opt} - \log c^* - \sigma^2/2}{\sigma} = \frac{\log(1/\alpha) - \log c^*}{2\sigma} = \frac{\mu\big(\log(1/\alpha) - \log c^*\big)}{2\Delta}.
	\end{equation*}
	By symmetry, the argument of the second $\Phi$ function is $-u$. Since $\Phi(-u) = 1 - \Phi(u)$, the maximum power improvement evaluates to:
\begin{equation*}
		G_{\max} = \Phi(u) - \Phi(-u) = 2\Phi\left( \frac{\mu\big(\log(1/\alpha) - \log c^*\big)}{2\Delta} \right) - 1.
	\end{equation*}
\end{proof}

\subsection{Proof of Theorem~\ref{thm:power_improvement_asymptotic}}

\begin{proof}
		The organizing idea is to convert the small-$\Delta$ problem into a large-$z$ problem: as $\Delta\to0^+$,  $z^*$ tends to $+\infty$ (we will verify this fact below), so after rewriting the relevant quantities in terms of $z^*$, it suffices to prove the corresponding limits for $z^*\to+\infty$. In the following proof, we use the notation $\sigma = \Delta/\mu$. As $\Delta\to0^+$, we have $\sigma\to0^+$. We first analyse the asymptotic behavior of $z^*$, which is defined by $\phi(z^*)/\Phi(z^*)=\sigma$. Let $h(z)=\phi(z)/\Phi(z)$. A direct differentiation gives $h'(z)=-h(z)\{z+h(z)\}$. Here $z+h(z)>0$: this is immediate for $z\ge0$, and for $z<0$ it follows from $\phi(z)=\int_{-\infty}^z(-u)\phi(u)\,du>-z\Phi(z)$. Hence $h$ is decreasing. Since $h(0)=\sqrt{2/\pi}$, whenever $\sigma<\sqrt{2/\pi}$ the solution to $h(z^*)=\sigma$ must satisfy $z^*>0$. Taking logarithms in the defining equation gives
		\[
			-\log\sigma=\frac{(z^*)^2}{2}+\log\sqrt{2\pi}+\log\Phi(z^*).
		\]
		For $z^*>0$, the term $\log\Phi(z^*)$ is bounded between $\log(1/2)$ and $0$. Hence the two terms after $(z^*)^2/2$ are negligible relative to $-\log\sigma$, and therefore $(z^*)^2\sim -2\log\sigma$, or equivalently $z^*\sim\sqrt{-2\log\sigma}$. In particular, $z^*\to+\infty$ as $\sigma\to0^+$. Therefore, for sufficiently small $\sigma$ we are in the first regime $\alpha \le \Phi(z^*)$, so $\log c^* = -\log\alpha + \log \Phi(z^*) - \sigma^2/2 - \sigma z^*$.
	By the law of total expectation, we have:
\begin{align*}
		\mathbb{P}_{H_{1}}\left(c^{*} \le E e^{-\xi} < \frac{1}{\alpha}\right) 
		&= \mathbb{E}\left(\mathbb{I}\left(c^{*} \le E e^{-\xi} < \frac{1}{\alpha}\right)\right) \\
		&= \int_{0}^{+\infty} \mathbb{P}\left(\log c^{*} \le \log x - \xi < \log(1/\alpha)\right) dF_{E}(x).
	\end{align*}
	Given $\xi \sim \mathcal{N}(\sigma^{2}/2, \sigma^{2})$, the probability in the integral can be written as the difference of two standard normal CDFs:
\begin{equation*}
		\mathbb{P}\left(\log c^{*} \le \log x - \xi < \log(1/\alpha)\right) = \Phi\left(\frac{\log x - \log c^{*} - \frac{\sigma^{2}}{2}}{\sigma}\right) - \Phi\left(\frac{\log x + \log \alpha - \frac{\sigma^{2}}{2}}{\sigma}\right).
	\end{equation*}
	Applying integration by parts we obtain:
\begin{align*}
		&\mathbb{P}_{H_{1}}\left(c^{*} \le E e^{-\xi} < \frac{1}{\alpha}\right) \\
		&= \left[ \left( \Phi\left(\frac{\log x - \log c^{*} - \frac{\sigma^{2}}{2}}{\sigma}\right) - \Phi\left(\frac{\log x + \log \alpha - \frac{\sigma^{2}}{2}}{\sigma}\right) \right) F_{E}(x) \right]_{0}^{+\infty} \\
		&\quad - \int_{0}^{+\infty} F_{E}(x) \left[ \phi\left(\frac{\log x - \log c^{*} - \frac{\sigma^{2}}{2}}{\sigma}\right) \frac{1}{\sigma x} - \phi\left(\frac{\log x + \log \alpha - \frac{\sigma^{2}}{2}}{\sigma}\right) \frac{1}{\sigma x} \right] dx \\
		&= \int_{0}^{+\infty} F_{E}(x) \phi\left(\frac{\log x + \log \alpha - \frac{\sigma^{2}}{2}}{\sigma}\right) \frac{1}{\sigma x} dx - \int_{0}^{+\infty} F_{E}(x) \phi\left(\frac{\log x - \log c^{*} - \frac{\sigma^{2}}{2}}{\sigma}\right) \frac{1}{\sigma x} dx.
	\end{align*}
		We apply the two changes of variables separately. For the first integral, set $t=(\log x+\log\alpha-\sigma^2/2)/\sigma$. Then $x=x_1(t)=\exp(\sigma t-\log\alpha+\sigma^2/2)$ and $dx=\sigma x_1(t)\,dt$. As $x$ ranges from $0$ to $+\infty$, $t$ ranges from $-\infty$ to $+\infty$, so
		\[
			\int_0^\infty
			F_E(x)\phi\left(\frac{\log x+\log\alpha-\sigma^2/2}{\sigma}\right)
			\frac{1}{\sigma x}\,dx
			=
			\int_{-\infty}^{+\infty}F_E(x_1(t))\phi(t)\,dt.
		\]
		
		For the second integral, set $t=(\log x-\log c^*-\sigma^2/2)/\sigma$. Then $x=x_2(t)=\exp(\sigma t+\log c^*+\sigma^2/2)$ and $dx=\sigma x_2(t)\,dt$. Hence
		\[
			\int_0^\infty
			F_E(x)\phi\left(\frac{\log x-\log c^*-\sigma^2/2}{\sigma}\right)
			\frac{1}{\sigma x}\,dx
			=
			\int_{-\infty}^{+\infty}F_E(x_2(t))\phi(t)\,dt.
		\]
		
		Combining the two transformed integrals, we obtain
		\[
			\mathbb{P}_{H_{1}}\left(c^{*} \le E e^{-\xi} < \frac{1}{\alpha}\right) 
			=
			\int_{-\infty}^{+\infty} \left[ F_{E}(x_{1}(t)) - F_{E}(x_{2}(t)) \right] \phi(t) dt.
		\]
	
		Using $c^*=\alpha^{-1}\Phi(z^*)\exp(-\sigma^2/2-\sigma z^*)$, the two arguments inside $F_E$ become $x_1(t)=\alpha^{-1}\exp(\sigma t+\sigma^2/2)$ and $x_2(t)=\alpha^{-1}\Phi(z^*)\exp\{\sigma(t-z^*)\}$. Therefore,
\begin{equation*}
		\mathbb{P}_{H_{1}}\left(c^{*} \le E e^{-\xi} < \frac{1}{\alpha}\right) = \int_{-\infty}^{+\infty} \left[ F_E\left( \frac{1}{\alpha} \exp\left(\sigma t + \frac{\sigma^2}{2}\right) \right) - F_E\left( \frac{1}{\alpha} \Phi(z^*) \exp(\sigma(t - z^*)) \right) \right] \phi(t) dt.
	\end{equation*}
	By replacing $\sigma$ with $\phi(z^*) / \Phi(z^*)$, we can express the integral in terms of $z^*$:
\begin{align*}
		&\mathbb{P}_{H_{1}}\left(c^{*} \le E e^{-\xi} < \frac{1}{\alpha}\right) \\
		&= \int_{-\infty}^{+\infty} \left[ F_E\left( \frac{1}{\alpha} \exp\left( t \frac{\phi(z^*)}{\Phi(z^*)} + \frac{1}{2}\left(\frac{\phi(z^*)}{\Phi(z^*)}\right)^2 \right) \right) - F_E\left( \frac{1}{\alpha} \Phi(z^*) \exp\left( (t - z^*) \frac{\phi(z^*)}{\Phi(z^*)} \right) \right) \right] \phi(t) dt.
	\end{align*}
		The preceding display writes the target probability as a function of $z^*$. Since $z^*\to+\infty$ as $\Delta\to0^+$, it remains to prove the corresponding generic limit:
\begin{equation*}
		\lim_{z \to +\infty} \dfrac{\int_{-\infty}^{+\infty} \left[ F_E\left( \frac{1}{\alpha} \exp\left( t \frac{\phi(z)}{\Phi(z)} + \frac{1}{2}\left(\frac{\phi(z)}{\Phi(z)}\right)^2 \right) \right) - F_E\left( \frac{1}{\alpha} \Phi(z) \exp\left( (t - z) \frac{\phi(z)}{\Phi(z)} \right) \right) \right] \phi(t) \,dt}{\frac{f_E(1/\alpha)}{\alpha} z \frac{\phi(z)}{\Phi(z)}} = 1.
	\end{equation*}
	Consider the integral
\begin{equation*}
		M_z = \int_{-\infty}^{+\infty} \left[ F_E\left( \frac{1}{\alpha} \exp\left(tr_z+\frac{1}{2}r_z^2\right) \right) - F_E\left( \frac{1}{\alpha}\Phi(z)\exp\left((t-z)r_z\right) \right) \right]\phi(t)\,dt.
	\end{equation*}
	For the region $|t|> z$, since
\begin{equation*}
		\int_{|t|>z}\phi(t)\,dt = O\left(\frac{\phi(z)}{z}\right) = o(zr_z),
	\end{equation*}
	the contribution to the integral from the region $|t|>z$ is $o(zr_z)$.
	
	For the region $|t|\le z$, we have $tr_z=O(zr_z)=o(1)$, $zr_z=o(1)$, and $r_z^2=o(1)$. Hence, applying the Taylor expansion uniformly for $|t|\le z$, we obtain
\begin{equation*}
		F_E\left( \frac{1}{\alpha} \exp\left(tr_z+\frac{1}{2}r_z^2\right) \right) = F_E\left(\frac{1}{\alpha}\right) + \frac{1}{\alpha}f_E\left(\frac{1}{\alpha}\right) \left[ \exp\left(tr_z+\frac{1}{2}r_z^2\right)-1 \right] + o(zr_z),
	\end{equation*}
	and similarly,
\begin{equation*}
		F_E\left( \frac{1}{\alpha}\Phi(z)\exp\left((t-z)r_z\right) \right) = F_E\left(\frac{1}{\alpha}\right) + \frac{1}{\alpha}f_E\left(\frac{1}{\alpha}\right) \left[ \Phi(z)\exp\left((t-z)r_z\right)-1 \right] + o(zr_z).
	\end{equation*}
	Therefore, the integral $M_z$ can be approximated as
\begin{equation*}
		M_z
=
\frac{1}{\alpha}f_E\left(\frac{1}{\alpha}\right)
\int_{-z}^{z}
\left[
e^{tr_z+r_z^2/2}
-
\Phi(z)e^{(t-z)r_z}
\right]\phi(t)\,dt
+
o(zr_z).
	\end{equation*}
It remains to extend the linearized integral from \([-z,z]\) to \(\mathbb{R}\). The added tail is bounded by a constant multiple of
\[
\int_{|t|>z}\left(e^{tr_z+r_z^2/2}+\Phi(z)e^{(t-z)r_z}\right)\phi(t)\,dt.
\]

We now show that the first term and second term are both $o(zr_z)$, so the added tail is negligible.

For the first term, note that
\(
e^{tr_z}\phi(t)
=
e^{r_z^2/2}\phi(t-r_z)
\).
Therefore, with the change of variables \(u=t-r_z\),
\[
\int_{|t|>z}e^{tr_z}\phi(t)\,dt
=
e^{r_z^2/2}\int_{|u+r_z|>z}\phi(u)\,du.
\]
If \(N\sim N(0,1)\), then the last integral equals
\[
e^{r_z^2/2}\mathbb{P}(|N+r_z|>z).
\]
Moreover,
\(\mathbb{P}(|N+r_z|>z)=\bar\Phi(z-r_z)+\Phi(-z-r_z)\), where
\(\bar\Phi(x)=1-\Phi(x)\). Since \(r_z=o(1)\), Mills' bound gives
\(\bar\Phi(z-r_z)\le \phi(z-r_z)/(z-r_z)=O(\phi(z)/z)\), where we used
\(\phi(z-r_z)=\phi(z)\exp\{zr_z-r_z^2/2\}=O(\phi(z))\). Similarly,
\(\Phi(-z-r_z)=\bar\Phi(z+r_z)\le \bar\Phi(z)=O(\phi(z)/z)\). Hence
\[
\mathbb{P}(|N+r_z|>z)=O\left(\frac{\phi(z)}{z}\right).
\]
So the first term satisfies
\[
\int_{|t|>z} e^{tr_z+r_z^2/2}\phi(t)\,dt
=
O\left(\frac{\phi(z)}{z}\right).
\]
Since \(\phi(z)/z=o(zr_z)\), this term is negligible.

The second term is treated identically, up to the bounded factor \(\Phi(z)e^{-zr_z}\). Since \(\phi(z)/z=o(zr_z)\), the added tail is negligible. Therefore,
\[
M_z
=
\frac{1}{\alpha}f_E\left(\frac{1}{\alpha}\right)
\int_{-\infty}^{\infty}
\left[
e^{tr_z+r_z^2/2}
-
\Phi(z)e^{(t-z)r_z}
\right]\phi(t)\,dt
+
o(zr_z).
\]

	Using the moment generating function of the standard normal distribution, we know that
\begin{equation*}
		\int_{-\infty}^{+\infty}e^{tr_z}\phi(t)\,dt = \exp\left(\frac{1}{2}r_z^2\right).
	\end{equation*}
	Hence,
\begin{equation*}
		\int_{-\infty}^{+\infty} \exp\left(tr_z+\frac{1}{2}r_z^2\right)\phi(t)\,dt = \exp(r_z^2),
	\end{equation*}
	and
\begin{equation*}
		\int_{-\infty}^{+\infty} \Phi(z)\exp\left((t-z)r_z\right)\phi(t)\,dt = \Phi(z)\exp\left(-zr_z+\frac{1}{2}r_z^2\right).
	\end{equation*}
	Hence, $M_z$ simplifies to
\begin{equation*}
		M_z = \frac{1}{\alpha}f_E\left(\frac{1}{\alpha}\right) \left[ \exp(r_z^2) - \Phi(z)\exp\left(-zr_z+\frac{1}{2}r_z^2\right) \right] + o(zr_z).
	\end{equation*}
	Since $1-\Phi(z) = O\left(\frac{\phi(z)}{z}\right) = O\left(\frac{r_z}{z}\right)$, we have
\begin{equation*}
		\Phi(z) = 1 + O\left(\frac{r_z}{z}\right).
	\end{equation*}
	Furthermore, using the Taylor expansions for the exponential functions,
\begin{equation*}
		\exp(r_z^2) = 1 + O(r_z^2),
	\end{equation*}
	and
Since \(1-\Phi(z)=O(\phi(z)/z)=O(r_z/z)\), we have
\(\Phi(z)=1+O(r_z/z)\). On the other hand,
\[
\exp\left(-zr_z+\frac{1}{2}r_z^2\right)
=
1-zr_z+O(z^2r_z^2)+O(r_z^2).
\]
Multiplying these two expansions gives
\[
\Phi(z)\exp\left(-zr_z+\frac{1}{2}r_z^2\right)
=
\{1+O(r_z/z)\}
\{1-zr_z+O(z^2r_z^2)+O(r_z^2)\}.
\]
The cross terms are absorbed into the displayed error terms, since
\(O(r_z/z)\cdot O(zr_z)=O(r_z^2)\). Hence
\[
\Phi(z)\exp\left(-zr_z+\frac{1}{2}r_z^2\right)
=
1-zr_z
+
O\left(\frac{r_z}{z}\right)
+
O(z^2r_z^2)
+
O(r_z^2).
\]
	It directly follows that the difference is given by
\begin{equation*}
		\exp(r_z^2) - \Phi(z)\exp\left(-zr_z+\frac{1}{2}r_z^2\right) = zr_z + o(zr_z).
	\end{equation*}
	Substituting this back, we obtain
\begin{equation*}
		M_z = \frac{1}{\alpha}f_E\left(\frac{1}{\alpha}\right)zr_z + o(zr_z).
	\end{equation*}
	Since $r_z = \frac{\phi(z)}{\Phi(z)}$, this can be rewritten as
\begin{equation*}
		M_z = \frac{f_E(1/\alpha)}{\alpha} z\frac{\phi(z)}{\Phi(z)} + o\left( z\frac{\phi(z)}{\Phi(z)} \right).
	\end{equation*}
	Therefore, we conclude that
\begin{equation*}
		\lim_{z\to+\infty}
		\frac{
			\int_{-\infty}^{+\infty} \left[ F_E\left( \frac{1}{\alpha} \exp\left( t\frac{\phi(z)}{\Phi(z)} + \frac{1}{2} \left(\frac{\phi(z)}{\Phi(z)}\right)^2 \right) \right) - F_E\left( \frac{1}{\alpha} \Phi(z) \exp\left( (t-z)\frac{\phi(z)}{\Phi(z)} \right) \right) \right]\phi(t)\,dt
		}{
			\frac{f_E(1/\alpha)}{\alpha} z\frac{\phi(z)}{\Phi(z)}
		}
		= 1.
	\end{equation*}
	Finally, because $r_{z^*}=\phi(z^*)/\Phi(z^*)=\sigma$ and $z^*\sim \sqrt{-2\log\sigma}$, we have
\[
	z^*r_{z^*}\sim \sigma\sqrt{-2\log\sigma}.
	\]
	Since $\mu$ is fixed and $\sigma=\Delta/\mu$, this is asymptotic to
\[
	\frac{\Delta}{\mu}\sqrt{-2\log\Delta}.
	\]
	Therefore,
\[
	\mathbb{P}_{H_1}\left(c^* \le E e^{-\xi}<\frac{1}{\alpha}\right)
	\sim \frac{f_E(1/\alpha)}{\alpha\mu}\Delta\sqrt{-2\log\Delta}.
	\]
\end{proof}

\subsection{Proof of Proposition~\ref{prop:gain}}
\begin{proof}
		The proof again reduces the low-sensitivity regime to a large-$z$ limit. Let $\sigma=\Delta/\mu$. As $\Delta\to0$, we have $\sigma\to0$ and the corresponding $z^*\to+\infty$ (as verified above), so we will rewrite the probability as a function of $z^*$ and prove the equivalent limit for a generic $z\to+\infty$. For sufficiently small $\sigma$ we use the first branch of $c^*$, so $\log c^*=-\log\alpha+\log\Phi(z^*)-\sigma^2/2-\sigma z^*$.
	
	Conditioning on the non-private $e$-value gives
	\[
		\mathbb{P}_{H_1}\left(Ee^{-\xi}\ge c^*, E<\frac{1}{\alpha}\right)
		=
		\int_0^{1/\alpha}
		\Phi\left(\frac{\log x-\log c^*-\sigma^2/2}{\sigma}\right)dF_E(x).
	\]
	An integration by parts yields
	\[
		\mathbb{P}_{H_1}\left(Ee^{-\xi}\ge c^*, E<\frac{1}{\alpha}\right)
		=
		F_E\left(\frac{1}{\alpha}\right)\Phi(t_0)
		-
		\int_0^{1/\alpha}F_E(x)\phi\left(\frac{\log x-\log c^*-\sigma^2/2}{\sigma}\right)\frac{1}{\sigma x}\,dx,
	\]
	where $t_0=(-\log\alpha-\log c^*-\sigma^2/2)/\sigma$.
	
		Apply the change-of-variables formula to the second integral with $t=(\log x-\log c^*-\sigma^2/2)/\sigma$. Equivalently, $x=x(t)=\exp(\sigma t+\log c^*+\sigma^2/2)$, so $dx=\sigma x(t)\,dt$ and $dx/(\sigma x)=dt$; the integration range becomes $(-\infty,t_0]$. Therefore,
	\[
		\mathbb{P}_{H_1}\left(Ee^{-\xi}\ge c^*, E<\frac{1}{\alpha}\right)
		=
		\int_{-\infty}^{t_0}
		\left[
		F_E\left(\frac{1}{\alpha}\right)-F_E(x(t))
		\right]\phi(t)\,dt.
	\]
		Using the first-branch formula for $c^*$, we have $t_0=z^*-\log\Phi(z^*)/\sigma$ and $x(t)=\alpha^{-1}\Phi(z^*)\exp\{\sigma(t-z^*)\}$. Write $r_z=\phi(z)/\Phi(z)$, so that $\sigma=r_{z^*}$. Then
		\[
			\mathbb{P}_{H_1}\left(Ee^{-\xi}\ge c^*, E<\frac{1}{\alpha}\right)
			=
			\int_{-\infty}^{z^*-\log\Phi(z^*)/r_{z^*}}
			\left[
			F_E\left(\frac{1}{\alpha}\right)
			-
			F_E\left(\frac{1}{\alpha}\Phi(z^*)\exp\{(t-z^*)r_{z^*}\}\right)
			\right]\phi(t)\,dt.
		\]
		This is now a function of $z^*$. To study its asymptotic behavior as $z^*\to+\infty$ (since as $\Delta \to 0$, $z^*\to+\infty$), for a generic $z$ define $t_0(z)=z-\log\Phi(z)/r_z$ and set
		\[
			M_z=
			\int_{-\infty}^{t_0(z)}
			\left[
			F_E\left(\frac{1}{\alpha}\right)
			-
			F_E\left(\frac{1}{\alpha}\Phi(z)\exp\{(t-z)r_z\}\right)
			\right]\phi(t)\,dt.
		\]
		The rest of the proof is the following asymptotic evaluation:
	\[
		M_z\sim \frac{f_E(1/\alpha)}{\alpha}zr_z,
		\qquad z\to+\infty.
	\]
		We now replace the upper endpoint $t_0(z)$ by $+\infty$ and check explicitly that this changes the integral only by $o(zr_z)$. Let
		\[
			A_z(t)=
			F_E\left(\frac{1}{\alpha}\right)
			-
			F_E\left( \frac{1}{\alpha}\Phi(z)\exp\{(t-z)r_z\} \right).
		\]
		Then
		\[
			M_z=\int_{-\infty}^{t_0(z)}A_z(t)\phi(t)\,dt,
			\qquad
			\int_{\mathbb R}A_z(t)\phi(t)\,dt
			=
			M_z+\int_{t_0(z)}^\infty A_z(t)\phi(t)\,dt.
		\]
		Since $F_E$ is a distribution function, $|A_z(t)|\le 1$. Also $t_0(z)>z$, because $\log\Phi(z)<0$ and $r_z>0$. Hence, by Mills' ratio,
		\[
			\left|\int_{t_0(z)}^\infty A_z(t)\phi(t)\,dt\right|
			\le
			\int_{t_0(z)}^\infty\phi(t)\,dt
			\le
			\int_z^\infty\phi(t)\,dt
			=O\left(\frac{\phi(z)}{z}\right)
			=o(zr_z),
		\]
		where the last step uses $r_z=\phi(z)/\Phi(z)$ and $\Phi(z)\to1$. Therefore,
		\begin{equation*}
			M_z = \int_{-\infty}^{+\infty} \left[ F_E\left(\frac{1}{\alpha}\right) - F_E\left( \frac{1}{\alpha}\Phi(z)\exp\left((t-z)r_z\right) \right) \right]\phi(t)\,dt + o(z r_z).
		\end{equation*}
			Write $B_z(t)=\Phi(z)\exp\{(t-z)r_z\}$. First split the full-line integral into the regions $|t|\le z$ and $|t|>z$. On the tail, the original integrand is bounded in absolute value by $1$, so
			\[
				\int_{|t|>z}
				\left|
				F_E\left(\frac{1}{\alpha}\right)
				-
				F_E\left(\frac{1}{\alpha}B_z(t)\right)
				\right|\phi(t)\,dt
				\le
				\int_{|t|>z}\phi(t)\,dt
				=
				o(zr_z).
			\]
			Hence
			\[
				M_z=
				\int_{|t|\le z}
				\left[
				F_E\left(\frac{1}{\alpha}\right)
				-
				F_E\left(\frac{1}{\alpha}B_z(t)\right)
				\right]\phi(t)\,dt
				+o(zr_z).
			\]
			On $|t|\le z$, we have $(t-z)r_z=O(zr_z)=o(1)$ uniformly. Hence $\exp\{(t-z)r_z\}=1+O(zr_z)$ uniformly on this region. Also $1-\Phi(z)=O(r_z/z)$, so $\Phi(z)=1+O(r_z/z)$. Since $r_z/z=O(zr_z)$, we get
			\[
				B_z(t)=\Phi(z)\exp\{(t-z)r_z\}
				=
				\{1+O(r_z/z)\}\{1+O(zr_z)\}
				=
				1+O(zr_z),
			\]
				uniformly for $|t|\le z$. Since $F_E$ is differentiable at $1/\alpha$ and
				\[
					\sup_{|t|\le z}|B_z(t)-1|=O(zr_z)\to0,
				\]
				there exists a remainder $R_z(t)$ satisfying $\sup_{|t|\le z}|R_z(t)|=o(zr_z)$ such that, uniformly for $|t|\le z$,
				\[
					F_E\left(\frac{1}{\alpha} B_z(t)\right)
					=
					F_E\left(\frac{1}{\alpha}\right)
					+
					\frac{1}{\alpha} f_E\left(\frac{1}{\alpha}\right)\{B_z(t)-1\}
					+
					R_z(t).
				\]
				After integration, the remainder contributes $o(zr_z)$, since $\int_{|t|\le z}\phi(t)\,dt\le 1$. Thus
			\[
				M_z=
				\frac{1}{\alpha}f_E\left(\frac{1}{\alpha}\right)
				\int_{|t|\le z}\{1-B_z(t)\}\phi(t)\,dt
				+
				o(zr_z).
			\]
			It remains to replace $|t|\le z$ by $\mathbb{R}$ in this last integral. The added tail is negligible because
			\[
				\int_{|t|>z}|1-B_z(t)|\phi(t)\,dt
				\le
				\int_{|t|>z}\phi(t)\,dt
				+
				\Phi(z)e^{-zr_z}\int_{|t|>z}e^{tr_z}\phi(t)\,dt.
			\]
		The first term is $O(\phi(z)/z)=o(zr_z)$. For the second term (actually we have proven this above), we spell out the exponential-tilting bound. Completing the square gives
		\[
			e^{tr_z}\phi(t)
			=
			\frac{1}{\sqrt{2\pi}}\exp\left(tr_z-\frac{t^2}{2}\right)
			=
			e^{r_z^2/2}\phi(t-r_z).
		\]
		Therefore, with the change of variables $u=t-r_z$,
		\[
			\int_{|t|>z}e^{tr_z}\phi(t)\,dt
			=
			e^{r_z^2/2}\int_{|u+r_z|>z}\phi(u)\,du
			=
			e^{r_z^2/2}\mathbb{P}(|N+r_z|>z),
			\qquad N\sim N(0,1).
		\]
		Moreover,
		\[
			\mathbb{P}(|N+r_z|>z)
			=
			\bar\Phi(z-r_z)+\Phi(-z-r_z),
		\]
		where $\bar\Phi(z) = 1-\Phi(z)$. Since $r_z\to0$ and $zr_z\to0$, Mills' ratio gives
		\[
			\bar\Phi(z-r_z)
			=
			O\left(\frac{\phi(z-r_z)}{z-r_z}\right)
			=
			O\left(\frac{\phi(z)}{z}\right),
			\qquad
			\Phi(-z-r_z)
			=
			O\left(\frac{\phi(z+r_z)}{z+r_z}\right)
			=
			O\left(\frac{\phi(z)}{z}\right).
		\]
		Here, for example, $\phi(z-r_z)=\phi(z)\exp\{zr_z-r_z^2/2\}=O\{\phi(z)\}$ and $z-r_z\asymp z$; the other tail is handled in the same way.
		Thus $\int_{|t|>z}e^{tr_z}\phi(t)\,dt=O\{\phi(z)/z\}$. Multiplying by the bounded factor $\Phi(z)e^{-zr_z}$ shows that the second term is also $o(zr_z)$. Therefore,
\begin{equation*}
		M_z = \frac{1}{\alpha}f_E\left(\frac{1}{\alpha}\right) \int_{-\infty}^{+\infty} \left[ 1 - \Phi(z)\exp\left((t-z)r_z\right) \right]\phi(t)\,dt + o(zr_z).
	\end{equation*}
	Using the moment-generating function of the standard normal distribution, the second term evaluates exactly as:
\begin{equation*}
		\int_{-\infty}^{+\infty} \Phi(z)\exp\left((t-z)r_z\right)\phi(t)\,dt = \Phi(z)\exp\left(-zr_z+\frac{1}{2}r_z^2\right).
	\end{equation*}
	Hence, $M_z$ simplifies to:
\begin{equation*}
		M_z = \frac{1}{\alpha}f_E\left(\frac{1}{\alpha}\right) \left[ 1 - \Phi(z)\exp\left(-zr_z+\frac{1}{2}r_z^2\right) \right] + o(zr_z).
	\end{equation*}
	Since $1-\Phi(z) = O(\phi(z)/z) = O(r_z/z)$, we know $\Phi(z) = 1 + O(r_z/z)$. Furthermore, applying the Taylor expansion for the exponential functions yields:
\begin{equation*}
		\exp\left(-zr_z+\frac{1}{2}r_z^2\right) = 1 - zr_z + O(z^2r_z^2) + O(r_z^2).
	\end{equation*}
	Multiplying these two components, the inner bracket expands to:
\begin{equation*}
		\Phi(z)\exp\left(-zr_z+\frac{1}{2}r_z^2\right) = 1 - zr_z + O\left(\frac{r_z}{z}\right) + O(z^2r_z^2) + O(r_z^2).
	\end{equation*}
	It directly follows that the difference is given by the leading term $zr_z$:
\begin{equation*}
		1 - \Phi(z)\exp\left(-zr_z+\frac{1}{2}r_z^2\right) = zr_z + o(zr_z).
	\end{equation*}
	Substituting this back into our expression for $M_z$, we obtain:
\begin{equation*}
		M_z = \frac{1}{\alpha}f_E\left(\frac{1}{\alpha}\right)zr_z + o(zr_z) = \frac{f_E(1/\alpha)}{\alpha} z\frac{\phi(z)}{\Phi(z)} + o\left( z\frac{\phi(z)}{\Phi(z)} \right).
	\end{equation*}
	Dividing by the denominator $\frac{f_E(1/\alpha)}{\alpha} z \frac{\phi(z)}{\Phi(z)}$, we conclude that:
\begin{equation*}
		\lim_{z\to+\infty}
		\frac{
			\int_{-\infty}^{t_0(z)} \left[ F_E\left(\frac{1}{\alpha}\right) - F_E\left( \frac{1}{\alpha} \Phi(z) \exp\left( (t-z)\frac{\phi(z)}{\Phi(z)} \right) \right) \right]\phi(t)\,dt
		}{
			\frac{f_E(1/\alpha)}{\alpha} z\frac{\phi(z)}{\Phi(z)}
		}
		= 1.
	\end{equation*}
	Since $z^*\phi(z^*)/\Phi(z^*) = z^*\sigma \sim \sigma\sqrt{-2\log\sigma}$ and $\sigma=\Delta/\mu$, the desired asymptotic expression follows:
\[
	\mathbb{P}_{H_1}\left(E e^{-\xi}\ge c^*, E<\frac{1}{\alpha}\right)
	\sim \frac{f_E(1/\alpha)}{\alpha\mu}\Delta\sqrt{-2\log\Delta}.
	\]
\end{proof}

\subsection{Proof of Proposition~\ref{prop:power_loss}}

\begin{proof}
		The proof follows the same reduction to a large-$z$ limit. Let $\sigma=\Delta/\mu$. As $\Delta\to0$, we have $\sigma\to0$ and the corresponding $z^*\to+\infty$ (as verified above), so we will rewrite the probability as a function of $z^*$ and prove the equivalent limit for a generic $z\to+\infty$. For sufficiently small $\sigma$ we use the first branch of $c^*$, so $\log c^*=-\log\alpha+\log\Phi(z^*)-\sigma^2/2-\sigma z^*$.
	
	Conditioning on $E$ and integrating by parts on $[1/\alpha,\infty)$ gives
	\[
		\mathbb{P}_{H_1}\left(E\ge\frac{1}{\alpha}, Ee^{-\xi}<c^*\right)
		=
		\int_{1/\alpha}^{+\infty}
		\left[
		F_E(x)-F_E\left(\frac{1}{\alpha}\right)
		\right]
		\phi\left(\frac{\log x-\log c^*-\sigma^2/2}{\sigma}\right)
		\frac{1}{\sigma x}\,dx.
	\]
		Apply the change-of-variables formula with $t=(\log x-\log c^*-\sigma^2/2)/\sigma$. Equivalently, $x=x(t)=\exp(\sigma t+\log c^*+\sigma^2/2)$, so $dx=\sigma x(t)\,dt$ and $dx/(\sigma x)=dt$. Using the first-branch formula for $c^*$ gives $x(t)=\alpha^{-1}\Phi(z^*)\exp\{\sigma(t-z^*)\}$, and the lower endpoint $x=1/\alpha$ becomes $t=z^*-\log\Phi(z^*)/\sigma$. Hence
\[
\mathbb{P}_{H_1}\left(E\ge\frac{1}{\alpha}, Ee^{-\xi}<c^*\right)
=
\int_{z^*-\log\Phi(z^*)/\sigma}^{\infty}
\left[
F_E\left(\frac{1}{\alpha}\Phi(z^*)e^{\sigma(t-z^*)}\right)
-
F_E\left(\frac{1}{\alpha}\right)
\right]
\phi(t)\,dt.
\]
	
			Write $r_z=\phi(z)/\Phi(z)$, so that $\sigma=r_{z^*}$. To study the preceding integral as a function of $z^*\to+\infty$, for a generic $z$ define $L_z=z-\log\Phi(z)/r_z$ and set
			\[
				M_z=
				\int_{L_z}^{+\infty}
				\left[
				F_E\left(\frac{1}{\alpha}\Phi(z)\exp\{(t-z)r_z\}\right)
			-
			F_E\left(\frac{1}{\alpha}\right)
					\right]\phi(t)\,dt.
				\]
				Then the displayed probability is exactly $M_{z^*}$. It remains to show that, as $z\to+\infty$,
			\[
				M_z
				\sim
				\frac{1}{e\alpha}f_E\left(\frac{1}{\alpha}\right)\frac{\phi^2(z)}{z^2}.
			\]
			After this is proved, we will translate the right-hand side back to $\sigma$ by using $r_{z^*}=\sigma$.

			By the definition of $L_z$, $\log\Phi(z)+r_z(L_z-z)=0$. Apply the change-of-variables formula with $t=L_z+s$, so $dt=ds$ and the range $t\in[L_z,\infty)$ becomes $s\in[0,\infty)$. The argument inside $F_E$ then becomes $\alpha^{-1}\exp(r_zs)$. Therefore,
\begin{equation*}
			M_z = \int_0^\infty \left[ F_E\left(\frac{1}{\alpha}e^{r_z s}\right) - F_E\left(\frac{1}{\alpha}\right) \right] \phi(L_z + s) \, ds.
	\end{equation*}
				Split the integral into $0\le s\le 1/(zr_z)$ and $s>1/(zr_z)$. We first show that the second range contributes only a negligible amount. The bracket in the integrand defining $M_z$ is bounded in absolute value by $1$. Thus, with $a_z=L_z+1/(zr_z)$, it is enough to show
			\[
				\int_{1/(zr_z)}^\infty \phi(L_z+s)\,ds
				=
				1-\Phi(a_z)
				=
				o\left(\frac{\phi^2(z)}{z^2}\right).
			\]
				Indeed, $r_z=\phi(z)/\Phi(z)\sim\phi(z)=(2\pi)^{-1/2}e^{-z^2/2}$, so $1/(zr_z)\sim \sqrt{2\pi}e^{z^2/2}/z$. Hence $1/(zr_z)$ is much larger than $z$, and $a_z/z\to+\infty$. For large $z$, we may take $a_z^2\ge 3z^2$, and Mills' ratio gives $1-\Phi(a_z)\le \phi(a_z)/a_z=O(e^{-3z^2/2})$. Since $\phi^2(z)/z^2\asymp e^{-z^2}/z^2$ and $e^{-3z^2/2}/(e^{-z^2}/z^2)=z^2e^{-z^2/2}\to0$, the displayed bound follows. Therefore,
				\[
					M_z
					=
					\int_0^{1/(zr_z)}
					\left[
					F_E\left(\frac{1}{\alpha}e^{r_z s}\right)
					-
					F_E\left(\frac{1}{\alpha}\right)
					\right]\phi(L_z+s)\,ds
					+
					o\left(\frac{\phi^2(z)}{z^2}\right).
				\]
				Now work on the truncated interval $0\le s\le1/(zr_z)$. On this interval, $0\le r_zs\le1/z$. Since $F_E$ is differentiable at $1/\alpha$ and $e^{r_zs}-1=r_zs\{1+o(1)\}$ uniformly, we have
	\begin{equation*}
					F_E\left(\frac{1}{\alpha}e^{r_z s}\right) - F_E\left(\frac{1}{\alpha}\right) = \frac{1}{\alpha}f_E\left(\frac{1}{\alpha}\right)r_z s + r_zs\,o(1),
				\end{equation*}
				where the $o(1)$ term is uniform for $0\le s\le1/(zr_z)$. Hence
				\[
					M_z
					=
					\frac{1}{\alpha}f_E\left(\frac{1}{\alpha}\right)r_z
					\int_0^{1/(zr_z)}s\phi(L_z+s)\,ds
					+
					o\left(r_z\int_0^{1/(zr_z)}s\phi(L_z+s)\,ds\right)
					+
					o\left(\frac{\phi^2(z)}{z^2}\right).
				\]
				The leading integral is currently over the truncated interval. To replace its upper limit by $\infty$, it remains to check that the added tail is negligible. This is where the variable $u$ enters: it is only the change of variables $u=L_z+s$. For large $z$, $L_z>0$, so
				\[
				\begin{aligned}
					r_z\int_{1/(zr_z)}^\infty s\phi(L_z+s)\,ds
					&=
					r_z\int_{a_z}^\infty (u-L_z)\phi(u)\,du
					\\
					&\le
					r_z\int_{a_z}^\infty u\phi(u)\,du
					\\
					&=
					r_z\phi(a_z)
					=
					o\left(\frac{\phi^2(z)}{z^2}\right).
				\end{aligned}
				\]
				The identity $\int_{a_z}^{\infty}u\phi(u)\,du=\phi(a_z)$ follows from $\phi'(u)=-u\phi(u)$. The last equality follows from the same bound above: $\phi(a_z)=O(e^{-3z^2/2})$, while $r_z\sim \phi(z)=O(e^{-z^2/2})$. Therefore $r_z\phi(a_z)=O(e^{-2z^2})=o\{\phi^2(z)/z^2\}$. Consequently,
	\begin{equation*}
				M_z = \frac{1}{\alpha}f_E\left(\frac{1}{\alpha}\right)r_z \int_0^\infty s\phi(L_z + s) \, ds + o\left(r_z\int_0^\infty s\phi(L_z+s)\,ds\right)+o\left(\frac{\phi^2(z)}{z^2}\right).
		\end{equation*}
	It remains to evaluate $\int_0^\infty s\phi(L_z + s) \, ds$. Set $u=L_z+s$, so that $s=u-L_z$ and the lower endpoint $s=0$ becomes $u=L_z$. Then
\begin{equation*}
	\begin{aligned}
		\int_0^\infty s\phi(L_z + s) \, ds
		&=
		\int_{L_z}^\infty (u-L_z)\phi(u)\,du \\
		&=
		\int_{L_z}^\infty u\phi(u)\,du
		-
		L_z\int_{L_z}^\infty \phi(u)\,du \\
		&=
		\phi(L_z)-L_z\{1-\Phi(L_z)\}.
	\end{aligned}
	\end{equation*}
	The last equality uses $\phi'(u)=-u\phi(u)$, so that $\int_{L_z}^{\infty}u\phi(u)\,du=\phi(L_z)$.
	
	We next analyse \(L_z\). From
\(L_z=z-\log\Phi(z)/r_z\), the expansion
\(\log(1-x)=-x+O(x^2)\), and Mills' expansion
\(1-\Phi(z)=\phi(z)\{z^{-1}-z^{-3}+O(z^{-5})\}\), we obtain
\[
-\log\Phi(z)
=
\phi(z)\left\{
\frac1z-\frac1{z^3}+O(z^{-5})
\right\}
+
O\left(\frac{\phi^2(z)}{z^2}\right).
\]
Since \(r_z=\phi(z)/\Phi(z)\sim\phi(z)\), dividing by \(r_z\) yields
\[
\frac{-\log\Phi(z)}{r_z}
=
\frac1z+O(z^{-3}).
\]
Substituting this back into the definition of \(L_z\), we conclude that
\[
L_z
=
z+\frac1z+O(z^{-3}).
\]

Using  identity
$
\int_0^\infty s\phi(L_z+s)\,ds
=
\phi(L_z)-L_z\{1-\Phi(L_z)\},
$
we expand the second term by Mills' formula:
\(1-\Phi(L_z)=\phi(L_z)\{L_z^{-1}-L_z^{-3}+O(L_z^{-5})\}\). Hence
\(L_z\{1-\Phi(L_z)\}=\phi(L_z)\{1-L_z^{-2}+O(L_z^{-4})\}\). Substituting this into the identity above, the leading \(\phi(L_z)\) terms cancel, and we obtain
\[
\int_0^\infty s\phi(L_z+s)\,ds
=
\frac{\phi(L_z)}{L_z^2}
\{1+O(L_z^{-2})\}.
\]

Moreover, from \(L_z=z+z^{-1}+O(z^{-3})\), we have
\(L_z^2=z^2+2+O(z^{-2})\). Hence
\[
\phi(L_z)
=
(2\pi)^{-1/2}e^{-L_z^2/2}
=
e^{-1}\phi(z)\exp\{O(z^{-2})\}
=
e^{-1}\phi(z)\{1+O(z^{-2})\}.
\]
Also,
\(
L_z^2=z^2+O(1)
=
z^2\{1+O(z^{-2})\}.
\)
Therefore,
$
L_z^{-2}
=
z^{-2}\{1+O(z^{-2})\}^{-1}.
$
Using the expansion \((1+x)^{-1}=1+O(x)\) as \(x\to0\), we obtain
\[
L_z^{-2}
=
z^{-2}\{1+O(z^{-2})\}.
\] Combining this with
\(\int_0^\infty s\phi(L_z+s)\,ds=\phi(L_z)L_z^{-2}\{1+O(L_z^{-2})\}\) and noting that \(L_z^{-2}=O(z^{-2})\), we obtain
\[
\int_0^\infty s\phi(L_z+s)\,ds
=
\frac{e^{-1}\phi(z)}{z^2}
\{1+O(z^{-2})\}.
\]

Combining with the preceding reduction
\[
M_z
=
\frac{1}{\alpha}f_E\left(\frac{1}{\alpha}\right)r_z
\int_0^\infty s\phi(L_z+s)\,ds
+
o\left(r_z\int_0^\infty s\phi(L_z+s)\,ds\right)
+
o\left(\frac{\phi^2(z)}{z^2}\right),
\]
and using \(r_z=\phi(z)/\Phi(z)\sim\phi(z)\), we obtain
\[
M_z
=
\frac{1}{e\alpha}f_E\left(\frac{1}{\alpha}\right)
\frac{\phi^2(z)}{z^2}
+
o\left(\frac{\phi^2(z)}{z^2}\right).
\]
Finally, since \(-\log r_z=\log\Phi(z)-\log\phi(z)=z^2/2+O(1)\), we have
\(r_z^2/(-\log r_z)\sim 2\phi^2(z)/z^2\). Therefore,
\[
M_z
\sim
\frac{f_E(1/\alpha)}{2e\alpha}
\frac{r_z^2}{-\log r_z}.
\]
Taking \(z=z^*\) gives \(r_{z^*}=\sigma\), and therefore
\[
	\mathbb{P}_{H_1}\left(E\ge\frac{1}{\alpha}, Ee^{-\xi}<c^*\right)
	\sim
	\frac{f_E(1/\alpha)}{2e\alpha}
	\frac{\sigma^2}{-\log\sigma}
	\sim
	\frac{f_E(1/\alpha)}{2e\alpha\mu^2}
	\frac{\Delta^2}{-\log\Delta}.
\]
\end{proof}

\subsection{Proof of Proposition~\ref{prop:general_composition}}
\begin{proof}
	By the composition property of Gaussian Differential Privacy, the combined release $(E_1^{\mathrm{DP}}, \dots, E_K^{\mathrm{DP}})$ satisfies $\sqrt{\sum_{k=1}^K \mu_k^2}$-GDP. Since the aggregation procedure $\mathcal{A}$ is applied only to the released private $e$-values, the aggregated output inherits the same privacy guarantee by post-processing. Its validity follows from the definition of a valid aggregation procedure. In particular, if $\mu_k=\mu$ for all $k$, the privacy guarantee becomes $\sqrt{K}\mu$-GDP.
\end{proof}

\subsection{Proof of Theorem~\ref{thm:independent_product}}
\begin{proof}
	By definition, $E_k^{\mathrm{DP}}(\mathcal{D}_k) = E_k(\mathcal{D}_k)\exp(-\xi_k)$, where $\xi_k \sim \mathcal{N}\big(\Delta_k^2/(2\mu^2), \Delta_k^2/\mu^2\big)$ are independent Gaussian noise variables. Since $\mathcal{D}_k$ and $\xi_k$ are independent, the expectation under the joint null hypothesis $H_0$ is:
\begin{equation*}
				\mathbb{E}_{H_0}[E_{\mathrm{prod}}^{\mathrm{DP}}(\mathcal{D})] = \prod_{k=1}^K \mathbb{E}_{H_0^{(k)}}[E_k(\mathcal{D}_k)] \mathbb{E}[e^{-\xi_k}] \le 1.
	\end{equation*}
			For privacy, let $M(\mathcal{D}) = \log E_{\mathrm{prod}}^{\mathrm{DP}}(\mathcal{D}) = \sum_{k=1}^K \log E_k(\mathcal{D}_k) - \sum_{k=1}^K \xi_k$.
	For any neighboring dataset $\mathcal{D} \sim \mathcal{D}'$, the disjoint partition ensures they differ in exactly one subset, say $\mathcal{D}_j \neq \mathcal{D}'_j$, while $\mathcal{D}_k = \mathcal{D}'_k$ for all $k \neq j$. The difference is bounded by:
\begin{equation*}
				\delta(\mathcal{D}, \mathcal{D}') =  \big|\sum_{k=1}^K \log E_k(\mathcal{D}_k) - \sum_{k=1}^K \log E_k(\mathcal{D}'_k) \big| = \big|\log E_j(\mathcal{D}_j) - \log E_j(\mathcal{D}'_j)\big| \le \max_{1 \le k \le K} \Delta_k.
	\end{equation*}
			The total noise $\sum_{k=1}^K \xi_k$ is Gaussian with variance $\sigma^2 = \sum_{k=1}^K \Delta_k^2/\mu^2$. Since $M(\mathcal{D})$ and $M(\mathcal{D}')$ are Gaussian with common variance $\sigma^2$ and mean difference $\delta(\mathcal{D}, \mathcal{D}')$, their trade-off function is:
\begin{equation*}
		T\big(M(\mathcal{D}), M(\mathcal{D}')\big) = G_{ \frac{\delta(\mathcal{D}, \mathcal{D}')}{\sigma} }.
	\end{equation*}
	Since $\delta(\mathcal{D}, \mathcal{D}')/\sigma \le (\max_k \Delta_k)/\sigma \triangleq \mu_{prod}$ and $\gamma \mapsto G_\gamma$ is decreasing, $T \ge G_{\mu_{prod}}$ for all neighboring dataset, satisfying $\mu_{prod}$-GDP.
	
	To prove the sharpness of $\mu_{prod}$, we show that no $\mu_2 < \mu_{prod}$ satisfies $\mu_2$-GDP for $M(\mathcal{D})$. Since the sensitivities $\Delta_k$ are suprema, they may not be exactly globally attainable. For any $\epsilon > 0$, let $j^* = \arg\max_k \Delta_k$, meaning $\Delta_{j^*} = \max_{k} \Delta_k$. There exists a pair of neighboring sub-datasets $\mathcal{D}_{j^*}^{(\epsilon)} \sim {\mathcal{D}'}_{j^*}^{(\epsilon)}$ such that:
\begin{equation*}
		\big| \log E_{j^*}(\mathcal{D}_{j^*}^{(\epsilon)}) - \log E_{j^*}({\mathcal{D}'}_{j^*}^{(\epsilon)}) \big| > \Delta_{j^*} - \epsilon.
	\end{equation*}
	Construct a global neighboring pair $\mathcal{D}^{(\epsilon)} \sim {\mathcal{D}'}^{(\epsilon)}$ matching this difference in the $j^*$-th partition while setting $\mathcal{D}_k = \mathcal{D}'_k$ for all $k \neq j^*$. The deterministic shift is:
\begin{equation*}
		\delta(\mathcal{D}^{(\epsilon)}, {\mathcal{D}'}^{(\epsilon)}) = \big| \log E_{j^*}(\mathcal{D}_{j^*}^{(\epsilon)}) - \log E_{j^*}({\mathcal{D}'}_{j^*}^{(\epsilon)}) \big| > \Delta_{j^*} - \epsilon.
	\end{equation*}
	The trade-off function for this pair is $G_{d(\epsilon)}$, where $d(\epsilon) = \delta(\mathcal{D}^{(\epsilon)}, {\mathcal{D}'}^{(\epsilon)})/\sigma > (\Delta_{j^*} - \epsilon)/\sigma$.
	Assume for contradiction that $M$ satisfies $\mu_2$-GDP for some $\mu_2 < \mu_{prod}$. This requires $G_{d(\epsilon)} \ge G_{\mu_2}$.
	Because $\mu_{prod} = \Delta_{j^*}/\sigma$ and $\mu_2 < \mu_{prod}$, choosing $\epsilon \in \big(0, \sigma(\mu_{prod} - \mu_2)\big)$ ensures $\mu_{prod} - \epsilon/\sigma > \mu_2$, giving $d(\epsilon) > \mu_2$. 
	Since $\gamma \mapsto G_\gamma$ is strictly decreasing, $d(\epsilon) > \mu_2$ implies $G_{d(\epsilon)} < G_{\mu_2}$, contradicting the $\mu_2$-GDP assumption. Thus, $\mu_{prod}$ is sharp.
\end{proof}

\subsection{Proof of Lemma~\ref{lem:e_max_extractor}}
\begin{proof}
	Algorithm~\ref{alg:e_max_extractor} involves two steps: selecting the index $j^*$ and releasing the noisy valid $e$-value. 
	By the general GDP composition theorem, if both steps satisfy $(\mu/\sqrt{2})$-GDP, their composition satisfies $\sqrt{(\mu/\sqrt{2})^2 + (\mu/\sqrt{2})^2}\text{-GDP} = \mu\text{-GDP}$. 
	The second step adds Gaussian noise $\xi \sim \mathcal{N}\big(\Delta^2/\mu^2, 2\Delta^2/\mu^2\big)$ to the logarithm. Since the sensitivity bound is $\Delta$, this satisfies $(\mu/\sqrt{2})$-GDP. We now prove the selection step satisfies $(\mu/\sqrt{2})$-GDP.
	
	Let $\mathcal{D}$ and $\mathcal{D}'$ be neighboring dataset, and $\theta_j(\mathcal{D}) = \log E_j(\mathcal{D})$. The sensitivity is $\sup_{\mathcal{D} \sim \mathcal{D}'} |\theta_j(\mathcal{D}) - \theta_j(\mathcal{D}')| \le \Delta$.
	The selection mechanism is $M(\mathcal{D}) = \arg\max_{i \in \mathcal{S}} (\theta_i(\mathcal{D}) + g_i)$, where $g_i \sim \text{Gumbel}(0, \beta)$ with $\beta = 2\Delta/\epsilon$. 
	Let $X_i(\mathcal{D}) = \theta_i(\mathcal{D}) + g_i$. The probability density and cumulative distribution functions of $g_i$ are $f_i(x) = (1/\beta) \exp\big(-(x - \theta_i(\mathcal{D}))/\beta - e^{-(x - \theta_i(\mathcal{D}))/\beta}\big)$ and $F_i(x) = \exp\big(-e^{-(x - \theta_i(\mathcal{D}))/\beta}\big)$.
	
	The probability of selecting index $i$ is:
\begin{align*}
		\mathbb{P}(M(\mathcal{D}) = i) &= \int_{\mathbb{R}} f_i(x) \prod_{j \neq i} F_j(x) \, dx \\
		&= \int_{\mathbb{R}} \frac{1}{\beta} \exp\left(-\frac{x-\theta_i(\mathcal{D})}{\beta}\right) \exp\left(- \sum_j e^{-\frac{x-\theta_j(\mathcal{D})}{\beta}}\right) \, dx \\
		&= \frac{\exp(\theta_i(\mathcal{D})/\beta)}{\sum_j \exp(\theta_j(\mathcal{D})/\beta)},
	\end{align*}
		To see the last equality, let $a_j=\exp(\theta_j(\mathcal{D})/\beta)$ and set $u=e^{-x/\beta}$. Then $dx=-\beta du/u$, $e^{-(x-\theta_j(\mathcal{D}))/\beta}=a_j u$, and the integral becomes $a_i\int_0^\infty \exp\{-u\sum_j a_j\}\,du=a_i/\sum_j a_j$.
		For neighboring dataset $\mathcal{D}, \mathcal{D}'$, the probability ratio is bounded by:
\begin{align*}
		\frac{\mathbb{P}(M(\mathcal{D}) = i)}{\mathbb{P}(M(\mathcal{D}') = i)} &= \frac{\exp(\theta_i(\mathcal{D})/\beta)}{\exp(\theta_i(\mathcal{D}')/\beta)} \frac{\sum_j \exp(\theta_j(\mathcal{D}')/\beta)}{\sum_j \exp(\theta_j(\mathcal{D})/\beta)} \\
		&\le \exp(\Delta/\beta) \cdot \exp(\Delta/\beta) = \exp\left(\frac{2\Delta}{\beta}\right) = e^{\epsilon}.
	\end{align*}
The inequality follows from the fact that $\exp(\theta_i(\mathcal{D})/\beta)\leq  \exp(\frac{\Delta}{\beta})\exp(\theta_i(\mathcal{D}')/\beta)$ and $\exp(\theta_i(\mathcal{D}')/\beta)\leq  \exp(\frac{\Delta}{\beta})\exp(\theta_i(\mathcal{D})/\beta)$.
	This implies $M$ satisfies $\epsilon$-DP, which is equivalent to $f_{\epsilon, 0}$-DP with trade-off function $f_{\epsilon, 0}(\alpha) = \max\{0, 1 - e^\epsilon \alpha, e^{-\epsilon}(1-\alpha)\}$.
	
	Choosing $\epsilon$ such that $e^\epsilon = \Phi(\mu/(2\sqrt{2}))/\Phi(-\mu/(2\sqrt{2}))$, the intersection of $1 - e^\epsilon \alpha$ and $e^{-\epsilon}(1-\alpha)$ occurs at:
\begin{equation*}
		\alpha^* = (1 - e^{-\epsilon})/(e^\epsilon - e^{-\epsilon}) = 1/(1 + e^\epsilon) = \Phi(-\mu/(2\sqrt{2})).
	\end{equation*}
	Evaluating the trade-off function at $\alpha^*$:
\begin{equation*}
		f_{\epsilon, 0}(\alpha^*) = 1 - e^\epsilon \alpha^* = 1 - \Phi(\mu/(2\sqrt{2})) = \Phi(-\mu/(2\sqrt{2})).
	\end{equation*}
	The GDP trade-off function $G_{\mu/\sqrt{2}}(\alpha) = \Phi\big(\Phi^{-1}(1-\alpha) - \mu/\sqrt{2}\big)$ evaluated at $\alpha^*$ gives:
\begin{equation*}
		G_{\mu/\sqrt{2}}(\alpha^*) = \Phi\big(\Phi^{-1}(\Phi(\mu/(2\sqrt{2}))) - \mu/\sqrt{2}\big) = \Phi(-\mu/(2\sqrt{2})).
	\end{equation*}
	The piecewise linear convex function $f_{\epsilon, 0}(\alpha)$ connects $(0,1)$, $(\alpha^*, \alpha^*)$, and $(1,0)$. The strictly convex curve $G_{\mu/\sqrt{2}}(\alpha)$ passes through these three points. Consequently, $f_{\epsilon, 0}(\alpha) \ge G_{\mu/\sqrt{2}}(\alpha)$ for all $\alpha \in [0,1]$. Thus, $M$ satisfies $(\mu/\sqrt{2})$-GDP.
\end{proof}

\subsection{Proof of Proposition~\ref{prop:peeling_privacy}}
\begin{proof}
			Algorithm~\ref{alg:e_peeling} executes the Report Noisy Max procedure (Algorithm~\ref{alg:e_max_extractor}) for $s$ adaptive iterations, each with privacy budget $\mu/\sqrt{s}$. By Lemma~\ref{lem:e_max_extractor}, each iteration satisfies $(\mu/\sqrt{s})$-GDP. By the general GDP composition theorem, the overall mechanism satisfies $\sqrt{s \times (\mu/\sqrt{s})^2}\text{-GDP} = \mu\text{-GDP}$.
			
		Now we prove the validity. For each $i \in \{1, \dots, m\}$, the output $\hat{E}_i$ is defined as $\tilde{E}_i$ if the hypothesis is selected and $0$ otherwise. Since $\tilde{E}_i \ge 0$ is a valid private $e$-value, the inequality $0 \le \hat{E}_i \le \tilde{E}_i$ holds by construction. Under the null hypothesis $H_i$, it follows that:
\begin{equation*}
		\mathbb{E}[\hat{E}_i \mid H_i] \le \mathbb{E}[\tilde{E}_i \mid H_i] \le 1.
	\end{equation*}
	Thus, $\hat{E}_i$ remains a valid $e$-value for all $i$, completing the proof.
	
	\end{proof}

\subsection{Proof of Corollary~\ref{cor:peeling_ebh_fdr}}
\begin{proof}
	By Proposition~\ref{prop:peeling_privacy}, each coordinate of $\widehat{\mathcal{E}}$ is a valid $e$-value. The e-BH procedure controls FDR under arbitrary dependence for any collection of valid $e$-values \citep{wang2022false,xu2025bringing}. Therefore, applying e-BH to $\widehat{\mathcal{E}}$ controls the FDR at level $\alpha$. The privacy statement follows from Proposition~\ref{prop:peeling_privacy} and the post-processing property of GDP.
\end{proof}

\subsection{Proof of Proposition~\ref{prop:adaptive_peeling}}
\begin{proof}
	We first establish the privacy cost of choosing $s$. If two log-$e$-value vectors differ by at most $\Delta$ in sup-norm, then each order statistic also differs by at most $\Delta$. Hence, for every $k\in\mathcal{K}$, the margin
\[
	Q_k=L_{(k)}-\log\{m/(\alpha k)\}
	\]
	has sensitivity at most $\Delta$. Adding Gaussian noise with variance $|\mathcal{K}|\Delta^2/\mu_0^2$ therefore makes each noisy margin $(\mu_0/\sqrt{|\mathcal{K}|})$-GDP. By the composition theorem for GDP, releasing the full vector $\{\widetilde Q_k:k\in\mathcal{K}\}$ satisfies $\mu_0$-GDP. The selected value $\hat{s}$ is a post-processing of these noisy margins.
	
	Conditional on the selected value $\hat{s}$, the final recursive peeling stage is Algorithm~\ref{alg:e_peeling} run with privacy budget $\mu_{\mathrm{peel}}=\sqrt{\mu^2-\mu_0^2}$. By Proposition~\ref{prop:peeling_privacy}, this stage satisfies $\mu_{\mathrm{peel}}$-GDP. Adaptive composition then gives
\[
	\sqrt{\mu_0^2+\mu_{\mathrm{peel}}^2}=\mu,
	\]
	so the full adaptive procedure satisfies $\mu$-GDP.
	
	It remains to verify e-validity. For any coordinate $i$, the output has the form
\[
	\hat E_i=I_i E_i e^{-\xi_i},
	\]
	where $I_i$ is the indicator that hypothesis $i$ is selected by the final peeling stage; if it is not selected, $I_i=0$. The value noise $\xi_i$ is drawn after selection and satisfies $\mathbb{E}[\exp(-\xi_i)\mid \hat{s}]=1$ for every realized value of $\hat{s}$. Therefore, under the null hypothesis for coordinate $i$,
\[
	\mathbb{E}[\hat E_i]
	= \mathbb{E}\!\left[I_i E_i\,\mathbb{E}\{\exp(-\xi_i)\mid \hat{s}, I_i, E_i\}\right]
	= \mathbb{E}[I_iE_i]
	\leq \mathbb{E}[E_i]\leq 1.
	\]
	Thus each coordinate of the output vector is a valid $e$-value.
\end{proof}

\section{Detailed Discussion on the Selection Mechanism: Gaussian vs. Gumbel }
\label{app:general_selection}

As briefly discussed in Remark~\ref{rmk:gaussian_selection}, the privacy bottleneck in adaptive peeling is the index-selection step. This issue is directly related to the proof of Lemma 7 in \citet{xia2023adaptive}, where the authors analyze a Gaussian report noisy minimum mechanism. In their notation, after transforming the $p$-values to scores $f_i(\mathcal{D})=G^{-1}(p_i(\mathcal{D}))$ for $i=1,\ldots,N$, the algorithm draws independent noises
\[
	Z_i\sim \mathcal{N}\left(0,\frac{8\Delta^2}{\mu^2}\right)
\]
and returns
\[
	j^*=\arg\min_i\{f_i(\mathcal{D})+Z_i\}.
\]
A claim central to the privacy guarantee of \cite{xia2023adaptive} is as follows.
\begin{claim}[Gaussian selection guarantee claimed in Lemma 7 of \citet{xia2023adaptive}]
	\label{claim:gaussian_selection_xia}
	The selection mechanism $\mathcal{D}\mapsto j^*$ above satisfies $(\mu/\sqrt{2})$-GDP.
\end{claim}

The distinction between report noisy minimum and report noisy maximum is immaterial for the issue studied here. Since a centered Gaussian distribution is symmetric, the noisy minimum rule in \citet{xia2023adaptive} can be rewritten, after changing signs, as the noisy maximum mechanism
\[
	\mathcal{M}(\mathcal{D})=\arg\max_i\big(v_i(\mathcal{D})+Z_i\big),
\]
where $v_i(\mathcal{D})=-f_i(\mathcal{D})=-G^{-1}\{p_i(\mathcal{D})\}$. We use this equivalent maximum formulation only for notational convenience. The counterexample below is constructed directly from mirror-conservative $p$-values satisfying the sensitivity condition in \citet{xia2023adaptive}. A parallel construction can also be used to show that Gaussian noise is insufficient for selection based on $\log E_i(\mathcal{D})$.

\subsection{A Counterexample to \ref{claim:gaussian_selection_xia}}
The factor $\sqrt{2}$ in Claim~\ref{claim:gaussian_selection_xia} is immaterial for the counterexample. Below, $\mu>0$ denotes the claimed GDP parameter for the selection step itself; replacing this symbol by $\mu/\sqrt{2}$ recovers the specific guarantee stated in Claim~\ref{claim:gaussian_selection_xia}. To show that $\mathcal{M}$ fails to satisfy $\mu$-GDP, it is sufficient to find a pair of neighboring datasets $\mathcal{D}, \mathcal{D}'$ and a decision rule $\phi$ such that the resulting Type II error violates the theoretical lower bound: $\beta_\phi < G_\mu(\alpha_\phi)$. %Hence $T(\alpha_\phi) \leq \beta_\phi < G_\mu(\alpha_\phi)$.

We first construct mirror-conservative $p$-value functions. As in \cite{xia2023adaptive}, we take $G=\Phi$, and let $\Delta>0$ denote the sensitivity. Fix any $0<\varepsilon<\Delta/2$. Let the null data-generating distribution put equal probability on four possible datasets,
\[
	\mathbb{P}(X=\mathcal{D})=\mathbb{P}(X=\mathcal{D}')=\mathbb{P}(X=\mathcal{E})=\mathbb{P}(X=\mathcal{E}')=\frac14,
\]
and regard every pair among these four datasets as neighboring (for instance, each dataset may be a singleton containing one of four possible records). Define $2N$ $p$-value functions $p_i(\cdot)$ by
\begin{align*}
	(p_1(\mathcal{D}),\ldots,p_{2N}(\mathcal{D}))
	&=(\underbrace{\Phi(-\varepsilon),\ldots,\Phi(-\varepsilon)}_{N},
	\underbrace{\Phi(-\frac{\Delta}{2}),\ldots,\Phi(-\frac{\Delta}{2})}_{N}),\\
	(p_1(\mathcal{D}'),\ldots,p_{2N}(\mathcal{D}'))
	&=(\underbrace{\Phi(-\frac{\Delta}{2}),\ldots,\Phi(-\frac{\Delta}{2})}_{N},
	\underbrace{\Phi(-\varepsilon),\ldots,\Phi(-\varepsilon)}_{N}),\\
	(p_1(\mathcal{E}),\ldots,p_{2N}(\mathcal{E}))
	&=(\underbrace{\Phi(\varepsilon),\ldots,\Phi(\varepsilon)}_{N},
	\underbrace{\Phi(\frac{\Delta}{2}),\ldots,\Phi(\frac{\Delta}{2})}_{N}),\\
	(p_1(\mathcal{E}'),\ldots,p_{2N}(\mathcal{E}'))
	&=(\underbrace{\Phi(\frac{\Delta}{2}),\ldots,\Phi(\frac{\Delta}{2})}_{N},
	\underbrace{\Phi(\varepsilon),\ldots,\Phi(\varepsilon)}_{N}).
\end{align*}

We now check the sensitivity assumption. For every coordinate $i$ and every pair $S,S'\in\{\mathcal{D},\mathcal{D}',\mathcal{E},\mathcal{E}'\}$,
\[
	\Phi^{-1}\{p_i(S)\},\Phi^{-1}\{p_i(S')\}\in
	\{-\frac{\Delta}{2},-\varepsilon,\varepsilon,\frac{\Delta}{2}\}.
\]
Hence
\[
	\left|\Phi^{-1}\{p_i(S)\}-\Phi^{-1}\{p_i(S')\}\right|\le \Delta
\]
for all six neighboring pairs. Thus the sensitivity condition is satisfied with the sensitivity $\Delta$.

We also check the mirror-conservative condition. For every coordinate $i$, the marginal distribution of $p_i(X)$ is uniform on
\[
	\{\Phi(-\frac{\Delta}{2}),\Phi(-\varepsilon),\Phi(\varepsilon),\Phi(\frac{\Delta}{2})\}.
\]
Thus, for any $0\le t_1\le t_2\le1/2$,
\[
	\mathbb{P}\{p_i(X)\in[t_1,t_2]\}
	=
	\mathbb{P}\{p_i(X)\in[1-t_2,1-t_1]\}.
\]
Each $p_i$ is therefore mirror-symmetric, and hence mirror-conservative.

Now convert this $p$-value construction to the score form used by the selection mechanism. The noisy-minimum rule of \citet{xia2023adaptive} uses scores $f_i=\Phi^{-1}(p_i)$; equivalently, after changing signs, it is a noisy-maximum rule with scores $v_i=-\Phi^{-1}(p_i)$. For notational convenience, we center these scores by defining
\[
	\widetilde v_i(S)=v_i(S)-\varepsilon
	=-\Phi^{-1}\{p_i(S)\}-\varepsilon.
\]
This centering does not change the argmax mechanism, because $\widetilde v_i(S)+Z_i$ differs from $v_i(S)+Z_i$ by the same constant for all coordinates $i$. Writing $\gamma=\frac{\Delta}{2}-\varepsilon>0$, the centered scores on the two neighboring datasets are
\begin{align*}
	\widetilde v(\mathcal{D})  &=(\underbrace{0,\ldots,0}_{N},\underbrace{\gamma,\ldots,\gamma}_{N}),\\
	\widetilde v(\mathcal{D}') &=(\underbrace{\gamma,\ldots,\gamma}_{N},\underbrace{0,\ldots,0}_{N}).
\end{align*}

An adversary observing the output index $j^* = \mathcal{M}(\cdot)$ applies the following symmetric test $\phi$: reject the null hypothesis ($H_0: $ the underlying dataset is $\mathcal{D}$) if the selected index falls in the first half, $j^* \in \{1,\dots,N\}$.

For this test, the Type I error $\alpha_\phi$ is the probability that the maximum of the $N$ ``loser'' candidates (true score $0$) exceeds the maximum of the $N$ ``winner'' candidates (true score $\gamma$). Due to the symmetry of the noise and the datasets, the Type II error probability $\beta_\phi$ is identical. Denote this probability as $p_{\text{error}}(N)$. We have:
$$ \alpha_\phi = \beta_\phi = p_{\text{error}}(N) = \mathbb{P}\Big( \max_{1 \le i \le N} Z_i > \gamma + \max_{N < j \le 2N} Z_j \Big). $$

For a valid $\mu$-GDP guarantee, the true trade-off function $f(\alpha) = \inf_{\alpha_\phi\leq \alpha} \beta_\phi$ must be bounded below by $G_\mu(\alpha)$ for all $\alpha$. By evaluating our specific test $\phi$, we establish an upper bound on the true trade-off at $\alpha = p_{\text{error}}(N)$:
$$ f\big(p_{\text{error}}(N)\big) \le \beta_\phi = p_{\text{error}}(N). $$
Thus, if we can show that $p_{\text{error}}(N) < G_\mu\big(p_{\text{error}}(N)\big)$ for some $N$, the mechanism is proven to violate $\mu$-GDP.

Assume $Z_i \sim \mathcal{N}(0, \sigma^2)$, as used in \citet{xia2023adaptive}, where $\sigma$ is fixed for the given sensitivity bound and claimed privacy parameter. According to extreme value theory for normal distributions \citep{leadbetter1983extremes}, the standard normal distribution belongs to the Gumbel domain of attraction. For standard normal variables $X_i \sim \mathcal{N}(0, 1)$, there exist normalizing sequences $a_N$ and $b_N$ given by:
\begin{equation*}
	a_N = \sqrt{2\log N} - \frac{\log(4\pi\log N)}{2\sqrt{2\log N}}, \quad b_N = \frac{1}{\sqrt{2\log N}},
\end{equation*}
such that $\frac{\max_{1 \le i \le N} X_i - a_N}{b_N} \xrightarrow{d} \Lambda$, where $\Lambda$ follows the standard Gumbel distribution.

For our scaled variables $Z_i = \sigma X_i$, we define the corresponding location and scale sequences as $\alpha_N = \sigma a_N$ and $\beta_N = \sigma b_N$. Construct the normalized maximums for the two independent groups:
\begin{align*}
	U_N &= \frac{\max_{1 \le i \le N} Z_i - \alpha_N}{\beta_N} \xrightarrow{d} U \sim \text{Gumbel}(0, 1) \\
	V_N &= \frac{\max_{N < j \le 2N} Z_j - \alpha_N}{\beta_N} \xrightarrow{d} V \sim \text{Gumbel}(0, 1)
\end{align*}
where $U$ and $V$ are independent standard Gumbel random variables. 

We can rewrite the error probability in terms of $U_N$ and $V_N$:
\begin{align*}
	p_{\text{error}}^{\text{Gauss}}(N) &= \mathbb{P}\Big( \beta_N U_N + \alpha_N > \beta_N V_N + \alpha_N + \gamma \Big) \\
	&= \mathbb{P}\Big( \beta_N (U_N - V_N) > \gamma \Big).
\end{align*}
By the continuous mapping theorem, the difference $W_N \triangleq U_N - V_N$ converges in distribution to $W \triangleq U - V$, which follows a standard Logistic distribution.

Observe the scale sequence $\beta_N = \frac{\sigma}{\sqrt{2\log N}}$. As $N \to \infty$, this deterministic sequence satisfies $\beta_N \to 0$. Thus, it also converges in probability: $\beta_N \xrightarrow{p} 0$.

Since $W_N \xrightarrow{d} W$ and $\beta_N \xrightarrow{p} 0$, Slutsky's theorem implies that:
\begin{equation*}
	\beta_N W_N \xrightarrow{d} 0 \cdot W = 0.
\end{equation*}
Convergence in distribution to a constant is equivalent to convergence in probability. Hence, for any $\eta>0$,
\begin{equation*}
	\lim_{N \to \infty} \mathbb{P}\Big( |\beta_N W_N| > \eta \Big) = 0.
\end{equation*}
Taking $\eta=\gamma$, where $\gamma$ is the fixed positive score gap, gives $\lim_{N \to \infty}\mathbb{P}\big(\beta_N W_N>\gamma\big)=0$.
Therefore, $\lim_{N \to \infty} p_{\text{error}}^{\text{Gauss}}(N) = 0$. Recall that our goal is to show $p_{\text{error}}(N) < G_\mu\big(p_{\text{error}}(N)\big)$. %{\color{red}This shows for sufficiently large $N$ it is possible to distinguish $\mathcal{D}$ and $\mathcal{D}'$ almost perfectly, violating the privacy guarantee.}
Since $\lim_{x \to 0}G_\mu(x)=1$, the fact that  $\lim_{N \to \infty} p_{\text{error}}^{\text{Gauss}}(N) = 0$ implies $p_{\text{error}}(N) < G_\mu\big(p_{\text{error}}(N)\big)$ for sufficiently large $N$.

%the $G_\mu$ curve decays much slower than the linear $y=x$ trajectory of our symmetric test. Therefore, for a sufficiently large $N$, $p_{\text{error}}^{\text{Gauss}}(N)$ falls below the $G_\mu$ lower bound.

\subsection{The Specific Gap in the Existing Proof}
\label{subsec:gaussian_flaw}

We now locate the specific step in the proof of Lemma 7 in \citet{xia2023adaptive} that fails. The authors attempt to establish $2^{-1/2}\mu$-GDP by first proving $(\epsilon, \delta(\epsilon))$-DP, where $\delta\left(\epsilon\right)=\Phi\left(-\epsilon/\mu+\mu/2\right)-e^{\epsilon}\Phi\left(-\epsilon/\mu-\mu/2\right)$. The problematic step is the passage from a bound for each fixed output index to a privacy guarantee for arbitrary output events.

Specifically, their argument obtains a bound of the form
\begin{equation} \label{eq:flawed_singleton}
	\mathbb{P}(\mathcal{M}(\mathcal{D}) = i) \le e^\epsilon \mathbb{P}(\mathcal{M}(\mathcal{D}') = i) + \delta(\epsilon)
\end{equation}
for each fixed index $i$. However, $(\epsilon, \delta)$-DP requires the inequality
\[
	\mathbb{P}(\mathcal{M}(\mathcal{D}) \in E) \le e^\epsilon \mathbb{P}(\mathcal{M}(\mathcal{D}') \in E) + \delta
\]
to hold uniformly for every subset $E$ of possible outputs.

If we sum the index-wise bounds over a subset $E$ with $|E| = k$, we only obtain
\begin{align*}
	\mathbb{P}(\mathcal{M}(\mathcal{D}) \in E) &= \sum_{i \in E} \mathbb{P}(\mathcal{M}(\mathcal{D}) = i) \\
	&\le \sum_{i \in E} \Big( e^\epsilon \mathbb{P}(\mathcal{M}(\mathcal{D}') = i) + \delta(\epsilon) \Big) \\
	&= e^\epsilon \mathbb{P}(\mathcal{M}(\mathcal{D}') \in E) + k \cdot \delta(\epsilon).
\end{align*}
The additive error accumulates with the size of the event $E$. Therefore, \eqref{eq:flawed_singleton} does not imply $(\epsilon, \delta(\epsilon))$-DP, and hence it also does not imply the claimed $2^{-1/2}\mu$-GDP guarantee for the selection step.

\subsection{Empirical Verification}
We also run a simulation to empirically verify privacy leakage. We consider a selection task among $2N = 2000000$ candidates. We take $\Delta=1$ and $\varepsilon=0.01$, so the centered score gap is $\gamma=\Delta/2-\varepsilon=0.49$. We construct a neighboring pair $(\mathcal{D}, \mathcal{D}')$ where $N$ candidates have centered score $\gamma$ and the other $N$ have centered score $0$. The groups are swapped between $\mathcal{D}$ and $\mathcal{D}'$. We fix the privacy parameter for the selection step at $\mu = 1/\sqrt{2}$. 

Following the construction in \citet{xia2023adaptive}, Gaussian noise is added to each candidate's score. We calculate the  error probability $p_{\text{error}}^{\text{Gauss}}$ through numerical integration of the overlap between the two extreme value distributions. This symmetric test yields an empirical anchor point $(\alpha, \beta) = (p_{\text{error}}^{\text{Gauss}}, p_{\text{error}}^{\text{Gauss}})$. We note that any valid trade-off function must be convex. Consequently, the true trade-off profile of the Gaussian mechanism is bounded above by the piecewise linear curve connecting $(0,1)$, this anchor point, and $(1,0)$. As shown in Figure~\ref{fig:violate}, this upper bound (red curve) falls into the violation zone below the theoretical $G_\mu$ lower bound (black curve), empirically confirming that the Gaussian mechanism fails to provide the claimed $\mu$-GDP guarantee for large candidate sets.
\begin{figure}
	\centering
	\includegraphics[width=0.6\textwidth]{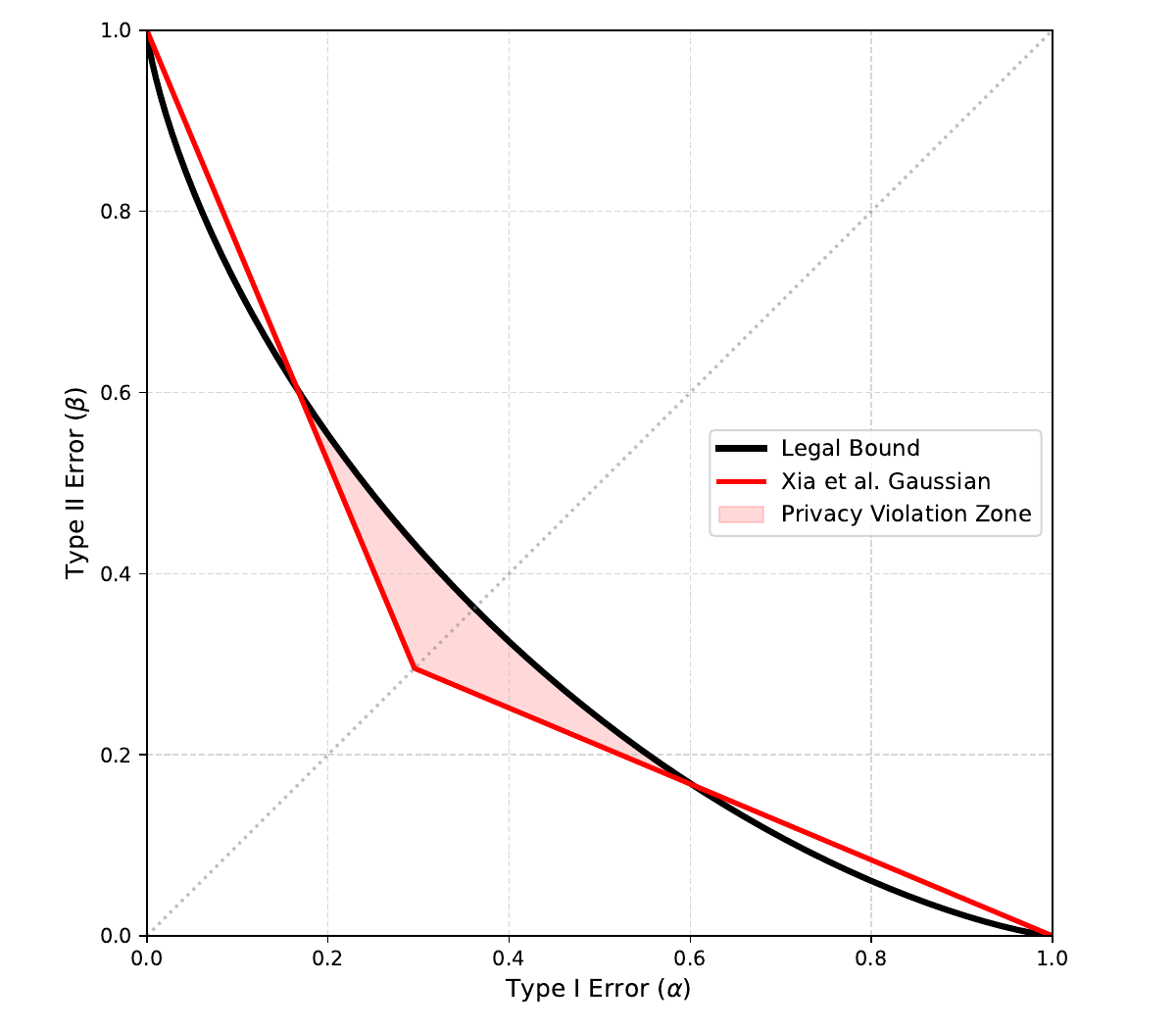}
	\caption{Comparison of the trade-off upper bound for the Gaussian selection mechanism (red) and $G_\mu$ (black). Under the specific setting, the red curve can fall below the black curve, which illustrates the technical difficulty of utilizing Gaussian noise for the $\mu$-GDP selection.}
	\label{fig:violate}
\end{figure}

\subsection{Cardinality-Independent Stability of Gumbel Noise}
To rectify this vulnerability, we employ Gumbel noise $Z_i \sim \text{Gumbel}(0, b)$. The fundamental advantage of the Gumbel distribution is its \textit{max-stability}. Unlike the Gaussian distribution, whose extreme values suffer from variance collapse, the maximum of $N$ independent $\text{Gumbel}(0, b)$ random variables has a Gumbel distribution with scale parameter $b$ and location shifted by $b \log N$.

Consequently, under our symmetric test $\phi$, the maximum scores for the ``loser'' and ``winner'' groups are exactly distributed as:
\begin{align*}
	M_A^{(N)} &= \max_{1 \le i \le N} Z_i \sim \text{Gumbel}(b \log N, b), \\
	M_B^{(N)} &= \max_{N < j \le 2N} (Z_j + \Delta) \sim \text{Gumbel}(\Delta + b \log N, b).
\end{align*}
Because the extreme values preserve the original distribution's dispersion, the fixed signal gap $\Delta$ is not dwarfed by a shrinking variance. Furthermore, the cardinality-dependent location shift $b \log N$ applies symmetrically to both groups. When evaluating the error probability $p_{\text{error}}^{\text{Gumbel}}(N) = \mathbb{P}(M_A^{(N)} > M_B^{(N)})$, the location shifts $b \log N$ cancel, yielding a closed-form logistic probability:
\begin{equation*}
	p_{\text{error}}^{\text{Gumbel}} = \frac{\exp((b \log N)/b)}{\exp((\Delta + b \log N)/b) + \exp((b \log N)/b)} = \frac{1}{1 + \exp(\frac{\Delta}{b})}.
\end{equation*}
The resulting  probability $p_{\text{error}}^{\text{Gumbel}}$ is independent of $N$. By calibrating the scale parameter $b = 2\Delta / \epsilon$ (derived from the desired selection budget $\mu$), the resulting $(\alpha_\phi, \beta_\phi)$ anchor point remains above the $G_{\mu}$ curve. 
    
\end{document}